\documentclass[reqno,11pt,letterpaper]{amsart}
\usepackage{amsmath,amssymb,amsthm,graphicx,mathrsfs,url}
\usepackage[usenames,dvipsnames]{color}
\usepackage[colorlinks=true,linkcolor=red,citecolor=Green]{hyperref}

\def\?[#1]{\textbf{[#1]}\marginpar{\Large{\textbf{??}}}}

\newtheorem{thm}{Theorem}
\newtheorem{prop}{Proposition}[section]
\newtheorem{defi}[prop]{Definition}

\newtheorem{ex}{Example}
\newtheorem{lemm}[prop]{Lemma}
\newtheorem{coro}[prop]{Corollary}

\numberwithin{equation}{section}

\newcommand{\nwc}{\newcommand}
\nwc{\ep}{\epsilon}
\nwc{\vareps}{\varepsilon}
\nwc{\Oph}{\operatorname{Op}_\hbar}
\nwc{\ra}{\rangle}

\nwc{\mf}{\mathbf} 
\nwc{\blds}{\boldsymbol} 
\nwc{\ml}{\mathcal} 

\nwc{\defeq}{\stackrel{\rm{def}}{=}}

\nwc{\cE}{\ml{E}}
\nwc{\cN}{\ml{N}}
\nwc{\cO}{\ml{O}}
\nwc{\cP}{\ml{P}}
\nwc{\cU}{\ml{U}}
\nwc{\cV}{\ml{V}}
\nwc{\cW}{\ml{W}}
\nwc{\tU}{\widetilde{U}}
\nwc{\IN}{\mathbb{N}}
\nwc{\IR}{\mathbb{R}}
\nwc{\IZ}{\mathbb{Z}}
\nwc{\IC}{\mathbb{C}}
\nwc{\IT}{\mathbb{T}}
\nwc{\IS}{\mathbb{S}}
\nwc{\tP}{\widetilde{P}}
\nwc{\tPi}{\widetilde{\Pi}}
\nwc{\tV}{\widetilde{V}}
\nwc{\rest}{\restriction}

\title{Gaussian free fields and Riemannian rigidity.}
\author{Nguyen Viet Dang}
\date{}
\begin{document}

\maketitle

\begin{abstract}
In the present paper, 
we show that on a compact Riemannian manifold $(M,g)$ 
of dimension $d\leqslant 4$, 
the renormalized partition function
$Z_g(\lambda)$ of a massive Gaussian Free Field 
determines the Laplace spectrum of $(M,g)$
hence imposes some strong geometric constraints on the Riemannian structure of $(M,g)$. 
In any finite dimensional family of Riemannian metrics of negative sectional curvature
bounded from below and above and whose isometry group is trivial, 
there is only a \textbf{finite number of isometry classes} of metrics with given partition function $Z_g(\lambda)$.
When $d<4$, the same result holds true if the random variable
$\int_M:\phi^2:dv$
has given probability distribution and without the lower bound on the sectional curvatures.   
\end{abstract}
%

\section{Introduction.}

In the present paper, we only consider smooth, compact, Riemannian manifolds
$(M,g)$ without boundary. For simplicity we also assume $M$ to be
connected and orientable.
On such manifold,
the Laplace--Beltrami operator $\Delta$ 
admits a discrete spectral resolution~\cite[Lemma 1.6.3 p.~51]{Gilkey} which means there is 
an increasing sequence
of eigenvalues~: $$\sigma (\Delta)=
\{0=\lambda_0< \lambda_1\leqslant \lambda_2 \leqslant \dots \leqslant \lambda_n\rightarrow +\infty \}$$
and corresponding $L^2$-basis of eigenfunctions
$(e_\lambda)_{\lambda\in \sigma(\Delta)}$ so that
$\Delta e_\lambda=\lambda e_\lambda$.

\subsubsection{Gaussian Free Fields and Feynman amplitudes.}

We next briefly recall 
the definition of the Gaussian free field (GFF)
associated to $\Delta$.
Our definition is probabilistic and
represents the Gaussian Free Field $\phi$
as a random distribution on $M$~\cite[Corollary 3.8 p.~21]{kang2017calculus}~\cite[eq (1.7) p.~3]{guillarmoupolyakov}~\cite{dimock2004markov} (see
also~\cite[section 4.2]{dubedatsle} for a related definition in a planar domain $D$).
In the classical physics litterature, this object is called Euclidean bosonic quantum field and can be defined
differently in terms of Gaussian measures on the space of distributions although the two definitions
are equivalent.

\begin{defi}[Gaussian Free Field]\label{dfn:green}
The Gaussian free field $\phi$ associated to
$(M,g)$ is defined as follows~:
denote by $(e_\lambda)_{\lambda\in \sigma(\Delta)}$ 
the spectral resolution
of $\Delta$. 
Consider a sequence $(c_\lambda)_{\lambda\in \sigma(\Delta)}, c_\lambda\in \mathcal{N}(0,1)$ 
of independent, identically distributed, centered Gaussian random variables.
Then we define
the Gaussian Free Field $\phi$ as the random series
\begin{equation}
\phi=\sum_{\lambda\in \sigma(\Delta)\setminus\{0\} } \frac{c_\lambda}{\sqrt{\lambda}}e_\lambda
\end{equation}
where the sum runs over the positive eigenvalues of
$\Delta$ and
the series converges almost surely as distribution in 
$\mathcal{D}^\prime(M)$.

The covariance of the Gaussian free field defined above
is the Green function~:
\begin{eqnarray*}
\boxed{\mathbf{G}(x,y)=\sum_{\lambda\in \sigma(\Delta)\setminus\{0\} } \frac{1}{\lambda}e_\lambda(x)e_\lambda(y)} 
\end{eqnarray*}
where the above series converges in $\mathcal{D}^\prime(M\times M)$.
\end{defi}

Note that in our definition of the Gaussian Free Field, we 
choose the random field $\phi$ to be orthogonal to constant functions so that the covariance of $\phi$ is exactly the Green's function
as defined above. The above means that the Gaussian measure $d\mu$ 
is constructed on the subspace $H^{s}(M)_0$, $\forall s<1-\frac{d}{2}$ 
of Sobolev
distributions 
orthogonal to constants.
This has the consequence that most 
of our arguments deal 
with the restriction of $\Delta$ 
to the orthogonal of constant functions.
Remark that to really construct a measure on $H^s(M)=H^s(M)_0\oplus \mathbb{R}$ would require tensoring $d\mu$ with 
the Lebesgue measure $dc$ on $\mathbb{R}$
which takes care of the zero modes~: $d\mu\otimes dc$~\cite[p.~30]{guillarmoupolyakov}. 
But we will not use this extended measure in the present paper.

We next recall the definition of polygon Feynman amplitudes.
\begin{defi}[Feynman amplitudes]\label{d:feynman}
Let $(M,g)$ be a closed compact Riemannian manifold and $\mathbf{G}$ the Green function of the
Laplace--Beltrami operator $\Delta$. 

For $n\geqslant 2$, define the formal product of Green function~:
\begin{equation}
t_n(x_1,\dots,x_n)=\mathbf{G}(x_1,x_2)\dots \mathbf{G}(x_n,x_1)
\end{equation}
as an element in $C^\infty(M^n\setminus \text{diagonals})$.
\end{defi}
The amplitude $t_n$ is well--defined outside diagonals 
because  the Green function
$\mathbf{G}$ is smooth outside the diagonal 
and develops singularities at coinciding points. 

In our main results, we will use the relation between these Feynman amplitudes and the
probability distribution of the Wick square of the GFF.

\subsubsection{From TQFT to Riemann invariants.}

In topological field theories of Chern--Simons~\cite{axelrodsinger2,kontsevich1994feynman,dijkgraaf1994perturbative} and
of BF type~\cite{baez2000introduction,witten19882+,carlip2003quantum}~\cite[3.4]{mnev2008discrete}, one has
a correspondence~:
\begin{eqnarray}
\boxed{\text{closed manifolds} \longrightarrow \text{partition function }Z\left(M\right)=\sum_{n=0}^\infty h^nF_n(M) }
\end{eqnarray}
where the $F_n\left(M \right)$ are invariants of the $C^\infty$-structure 
and do not depend on the choice of metrics needed to define the propagator of the theory.
For interacting scalar quantum
field on $\mathbb{R}^4$, 
it was proved by Belkale--Brosnan~\cite{belkale2003periods} and 
Bogner--Weinzierl~\cite{Bogner2009periods} that
Feynman amplitudes are special numbers called periods. On general manifolds,
as a consequence of the quantum field theory formalism of Segal~\cite{segal2004definition}, Stolz--Teichner~\cite{stolz2004elliptic,kandel2015functorial,stolz2014lecture}, a QFT is expected to
give a correspondence 
from closed
manifolds endowed with extra structure, for instance Riemannian, complex structures or vector bundles over $M$, to
the complex numbers
$\mathbb{C}$ or the ring of formal power series over $\mathbb{C}$.
On Riemannian manifolds, the numbers of QFT might become sensitive to variations of the metric $g$ 
and have no reasons to be periods anymore. The goal of the present paper is to study the dependence 
of the partition
function $Z_g$ and Feynman amplitudes on the Riemannian metric $g$.

\subsubsection{Probing Riemannian geometries with quantum fields.}

The study of Euclidean quantum fields has a long history in the constructive quantum field theory community with seminal contributions
of Albeverio, Fr\"ohlich, Gallavotti, Glimm, Guerra, Jaffe, Nelson, Seiler,
Spencer, Simon, Symanzik and Wightman just to name a few, see~\cite{Glimm, simonpphi2, simonfunct, FrohlichSokal} and the references inside. Our goal in the present paper, is to relate 
the properties of the quantum field on the manifold with the geometric properties of the underlying manifold itself.
We were inspired in part by the work of Seiler~\cite{Seilergauge} 
who stressed the relation between quantum fields and functional
determinants.

At this point, we should stress that our results 
on rigidity of Riemannian structures come in two flavors~:
\begin{enumerate}
\item the diffeomorphism type of $M$ (topology and $C^\infty$ structure) is fixed and the question is about the metric $g$ up to isometry,
\item the diffeomorphism type of $M$ is not fixed and the question is about the \textbf{pair} $(M,g)$ up to isometry.
\end{enumerate}

In what follows, we denote by $\mathbb{C}[[\lambda]]$ the ring of formal power series in $\lambda$.
We shall study the
renormalized partition function 
\begin{equation}
Z_g\left(\lambda \right)=\mathbb{E}\left(\exp\left( -\frac{\lambda}{2}\int_M :\phi^2(x):dv \right) \right)\in \mathbb{C}[[\lambda]]
\end{equation}
of a free bosonic theory when the dimension of $(M,g)$ equals $2\leqslant d\leqslant 4$ where we need some extra renormalization 
when $d=4$, see Proposition~\ref{t:mainthm2}. 
The partition function $Z_g(\lambda)\in \mathbb{C}[[\lambda]]$
depends only on the isometry class of $(M,g)$ and a natural question would be what informations on $(M,g)$ can be extracted from
$Z_g\left(\lambda \right)$ as formal power series. 

 We can also formulate a related question as follows~: 
let $\varphi:M\mapsto M$ be
a diffeomorphism and let $\mathbf{G}$ be the Green function of $\Delta_g$, where the notation $\Delta_g$ is used
to stress the dependence on the metric $g$.
$\left(M,\varphi^*g\right)$ is \textbf{isometric} to $(M,g)$ and induces a diffeomorphism
$\Phi:M\times M\mapsto  M\times M$ such that  
the pulled--back Green function $\Phi^*\mathbf{G}\in \mathcal{D}^\prime(M\times M)$
is the Green function of the  
Laplace--Beltrami
operator $\Delta_{\Phi^*g}$ of the pulled--back metric.
It follows that integrals of non divergent Feynman amplitudes 
associated to closed graphs are 
\textbf{isometry invariant numbers} and
depend only on 
the Riemannian structure of the pair $(M,g)$.
In the terminology of subsubsection~\ref{ss:modulimetrics}, we will say they 
induce functions
on the \emph{moduli space of metrics}.
What 
informations on the Riemannian structure $(M,g)$ can be recovered
from integrals of Feynman amplitudes over configuration space ?

\section{Main results.}

In what follows, we 
introduce some preliminary definitions on convergence of sequences of Riemannian manifolds and moduli spaces of metrics that we need to
state our two main results on Riemannian rigidity from quantum fields.

\subsubsection{Convergent sequences of Riemannian manifolds in Lipschitz topology.}

In the present paragraph, the diffeomorphism 
type of the manifolds is \textbf{not fixed}.
Let us recall that
the set of $C^\infty$ Riemannian metrics on $M$ with the usual Fr\'echet topology on smooth
$2$-tensors
is denoted by $\textbf{Met}(M)$. 
It is an open convex cone of the space of
symmetric $2$--tensors in the $C^\infty$ topology. 
We have the natural action of $\mathbf{Diff}(M)$, the set of 
diffeomorphisms of $M$ acting by pull--back on 
$\textbf{Met}(M)$. 
Two Riemannian manifolds $(M_1,g_1)$ and $(M_2,g_2)$ are equivalent if
there exists a $C^\infty$ diffeomorphism 
$\varphi:M_1\mapsto M_2$ s.t. $\varphi^*g_2=g_1$.
Then $\mathfrak{Riemm}$ is the set of equivalence classes of 
Riemannian manifolds. 
We insist that elements in $\mathfrak{Riemm}$ can have \textbf{different 
diffeomorphism types}.
In our work, we will only need to define Lipschitz convergence for sequences of smooth Riemannian manifolds. 
Concretely,
a sequence of isometry classes of Riemannian manifolds 
$(M_n,g_n)_{n\in \mathbb{N}}$ converges to $(M,g)$ in $\mathfrak{Riemm}$ 
for the Lipschitz topology if
there exists a sequence $\varphi_n~: M_n\mapsto M$ of 
bilipschitz homeomorphisms
s.t. both $\sup_{x\in M_n} \Vert d\varphi_n\Vert$
and $\sup_{x\in M_1} \Vert d\varphi_n^{-1}\Vert$ 
tend to $1$ when $n\rightarrow +\infty$.
We will need the Lipschitz topology in the formulation of our main Theorem~\ref{t:mainthm2}
and also when we discuss compactness properties of isospectral metrics in subsubsection~\ref{ss:compactness}.

\subsubsection{The moduli space of metrics.}
\label{ss:modulimetrics}

In this paragraph, we fix the smooth manifold $M$ and only the metrics on $M$ will vary.
We define the \textbf{moduli space of Riemannian metrics}
as a quotient space~\cite[p.~381]{BergerEbin}~:
\begin{eqnarray}
\boxed{\mathcal{R}(M)=\textbf{Met}(M)/\mathbf{Diff}(M)}
\end{eqnarray}
endowed with the quotient topology.
In practice, 
a sequence of isometry classes $[g_n]\underset{n \rightarrow +\infty}{\rightarrow} [g]$ if there is a sequence of representatives $g_n$ of $[g_n]$ which converges to $g$ in the $C^\infty$--topology~\cite[p.~602]{sarnak1990determinants}~\cite[p.~233]{BGM} (see also~\cite[p.~175]{berger2012panoramic}). 
For every $0<\varepsilon<1$, we use the notations $\mathcal{R}(M)_{\leqslant-\varepsilon}$
and $\mathcal{R}(M)_{[-\varepsilon^{-1},-\varepsilon]}$ for the 
moduli space of Riemannian metrics with negative
sectional curvatures bounded from above by $-\varepsilon$ and 
whose sectional curvature is contained in 
$[-\varepsilon^{-1},-\varepsilon]$ respectively. 

It is a result of Ebin~\cite{Ebin1, Ebinthesis} that $\mathcal{R}(M)$ endowed with the quotient topology is a  
Hausdorff \textbf{metric space}~\cite[p.~317--319]{Fischer}. In the sequel, we shall summarize the main properties 
of the metric structure on $\mathcal{R}(M)$.

\subsubsection{The moduli space $\mathcal{R}(M)$ as a metric space.}

For every $s>\frac{\dim(M)}{2}$, Ebin considered the Hilbert manifold $\textbf{Met}^s(M)$ of Sobolev metrics of regularity $s$ and
the topological group $\textbf{Diff}^{s+1}(M)$ of bijective maps $f$ s.t.
both $f$ and $f^{-1}$ are Sobolev maps in $H^{s+1}(M,M)$~\cite[2.3 p.~158]{MarsdenEbinFischer} acting on $\textbf{Met}^s(M)$.
He constructed a Riemannian metric $\textbf{g}_s$ on $\textbf{Met}^s(M)$, called Sobolev metric of degree $s$, which is invariant by the action of $\textbf{Diff}^{s+1}(M)$ and is defined as follows. 
The tangent space $T_g\textbf{Met}^s(M)$ at $g\in \textbf{Met}^s(M)$ 
is naturally identified with the Sobolev space $H^s(S^2T^*M)$ of Sobolev sections
of $S^2T^*M$ of regularity $s$. So for every $h\in H^s(S^2T^*M)\simeq T_g\textbf{Met}^s(M)$,
\begin{equation}
\left\langle h,h \right\rangle_{\textbf{g}_s}=\sum_{k=0}^s\int_M \left\langle\nabla_g^kh, \nabla_g^k h\right\rangle_{S^{k+2}T^*M} dv_g 
\end{equation}
where $dv_g$ is the volume form induced by $g$, $\nabla_g$ is the 
covariant derivative defined by $g$ acting on $H^s(S^2T^*M)$ and $ \left\langle .,.\right\rangle_{S^{k+2}T^*M} $ denotes the fiberwise scalar product on the bundle $S^{k+2}T^*M$ induced by $g$. 
The corresponding distance function 
on $\textbf{Met}^s(M)$
is denoted by $\textbf{d}^s$.
Following Fischer~\cite[p.~319]{Fischer}, we may define a distance $\textbf{d}$ on $\textbf{Met}(M)$ as
follows~:
\begin{equation}
\textbf{d}(g_1,g_2)=\sum_{k>\frac{\dim(M)}{2}} \frac{1}{2^k} \frac{\textbf{d}^k(g_1,g_2) }{1+\textbf{d}^k(g_1,g_2)}
\end{equation}
where the distance $\textbf{d}$ is $\mathbf{Diff}(M)$ invariant by construction. 
Hence $\textbf{d}$ induces a distance on the quotient
space $\mathcal{R}(M)$ which
generates the quotient topology.
Now that we recalled the metric space 
structure of $\mathcal{R}(M)$, a natural question is if
$\mathcal{R}(M)$ admits a smooth manifold structure.

\subsubsection{The regular part $\mathcal{G}$ of $\mathcal{R}(M)$.}
Unfortunately, the answer is negative and $\mathcal{R}(M)$ should be understood
as some kind of orbifold. 
The proper setting for the analysis in infinite dimensional space of metrics is that of
inductive limit of Hilbert spaces (ILH)
structures
which are 
specializations of Fr\'echet manifolds defined by 
Omori~\cite[Def II.5 p.~4]{Bourguignon}. 
So all the words \emph{submanifolds} or \emph{diffeomorphisms}
must be understood in the sense of ILH submanifolds and diffeomorphisms.
The set $\mathcal{R}(M)$ does not have a manifold structure but it is a fundamental result of
Ebin~\cite{Ebin1, Ebinthesis} and Palais independently that the action of $\mathbf{Diff}(M) $
on $\textbf{Met}(M)$ admits slices.
Moreover Ebin~\cite{Ebin1, Ebinthesis} proved that for a metric $g$ whose 
isometry group $I_g$ is trivial, $I_g=\{ \varphi\in \mathbf{Diff}(M), \varphi^*g=g  \}=Id$, 
the quotient $\mathcal{R}(M)$ has a manifold structure in some neighborhood of $[g]$. 
Such $[g]$ are called \textbf{regular points} of the moduli space $\mathcal{R}(M)$ and the set of regular points is denoted 
by $\mathcal{G}$.
It is a result of Ebin that $\mathcal{G}\subset \mathcal{R}(M)$ is \textbf{open dense and has
a smooth manifold structure}.
In the sequel, for every $\varepsilon>0$, we shall denote by
$\mathcal{G}_{\geqslant\varepsilon}$ the set of $[g]\in \mathcal{G}$ s.t.
$\textbf{d}([g],\partial\mathcal{G})\geqslant\varepsilon$ where $\partial\mathcal{G}=
\mathcal{R}(M)\setminus\mathcal{G} $. 


\subsubsection{Fluctuations of the integrated Wick square.}

In quantum field theory on curved space times, one is interested
in the behaviour of the stress--energy tensor and its fluctuations
under quantization of the 
fields assuming that the metric stays classical. 
For instance, many works of Moretti~\cite{moretti1999,moretti1999local,moretti2003,moretti2011,hack2012stress} deal
with the renormalization of various quantum field theoretic quantities, for instance the stress--energy tensor, 
using zeta regularization and local point splitting methods. 
In the present paper, 
we
study fluctuations of the integral of the 
Wick square $\int_M :\phi^2(x):dv$ on the manifold $M$ which is a simpler observable and is the integral of the 
\emph{field fluctuations}
in Moretti's work. 
In aQFT, it also appears 
in the work of Sanders~\cite{sanders2017local} and is
interpreted as a local temperature. 

In probability, 
the Wick square is also 
related to 
loop measures associated to some 
random walks 
on graphs~\cite{leJan,lawler2018topics} and it is assumed that the continuous Wick square should be related
to some loop measures.  
On Riemannian manifolds of negative curvature,
there is a strong
relation between Brownian motion on the base manifold $M$, the continuous version of random walks, 
and the geodesic flow on the unitary cosphere bundle $S^*M$ over $M$ which is Anosov. This topic
was studied by
many authors like Ancona,
Arnaudon, Guivarc'h, Kaimanovich, 
Kendall, Kifer, Ledrappier, Le Jan, Pinsky and Thalmaier among many others (see~\cite{arnaudonthalm} and references therein).
Our main results, Proposition~\ref{t:mainthm2} and Theorem~\ref{t:mainthm1},    
give 
an explicit relation
between fluctuations of the Wick squares, the partition function $Z_g(\lambda)$,
periodic geodesics and rigidity on manifolds with negative curvature.

\subsubsection{Periods of the geodesic flow.}

We recall the definition of the periods of
the geodesic flow~\cite[section 10.5]{berger2012panoramic}. 
\begin{defi}[Periods]
\label{d:periods}
Let us consider the moduli space of Riemannian metrics $\mathcal{R}(M)$ on $M$. 
For every element of $\mathcal{R}(M)$, choose a representative $g$. 
We denote by $(\Phi^t)_g:S^*M\mapsto S^*M$ the geodesic flow acting on
the unitary cosphere bundle $S^*M$.
Then for every class $[g]\in \mathcal{R}(M)$, 
we define the \textbf{periods}
$\mathcal{P}([g])$ as the set~:
\begin{equation}
\boxed{\mathcal{P}([g])=\{ T>0 \text{ s.t. } \Phi^T_g(x;\xi)=(x;\xi) \text{ for some }(x;\xi)\in S^*M \}\subset \mathbb{R}_{>0}.}
\end{equation}
\end{defi}
The set $\mathcal{P}(g)$ is called the \textbf{length spectrum} of $(M,g)$.
\subsection{Main results.}
Recall we defined the formal product $t_n$ of Green functions in definition~\ref{d:feynman}.
For a compact operator $A$,
we will denote by $\sigma(A)$ the set of singular values of $A$.
On any oriented smooth manifold $X$, we shall denote by $\vert\Lambda^{top}\vert X$
the bundle of densities on $X$ and by $C^\infty\left(\vert\Lambda^{top}\vert X\right)$ its smooth sections. 
Our first result reads~:
\begin{prop}\label{t:mainthm2}
Given a closed compact Riemannian manifold $(M,g)$ of dimension $2\leqslant d\leqslant 4$, a function $V\in C^\infty(M)$, define  
the sequence of numbers
$$c_n(g,V)=\int_{M^n} t_{n}(x_1,\dots,x_n)V(x_1)\dots V(x_n) dv_n,$$ $n\in \mathbb{N}$
where $dv_n$ is the Riemannian density
in $C^\infty\left(\vert\Lambda^{top}\vert M^n\right)$.
For $\varepsilon\in (0,1]$,
let $\phi_\varepsilon=e^{-\varepsilon\Delta}\phi$ be the
heat regularized GFF, $:\phi^2_\varepsilon(x):=\phi^2_\varepsilon(x)-\mathbb{E}\left(\phi^2_\varepsilon(x) \right)$ and
define the renormalized partition functions~: 
\begin{eqnarray*}
Z_g(\lambda,V)&=&\lim_{\varepsilon\rightarrow 0^+}\mathbb{E}\left(\exp\left(-\frac{\lambda}{2}\int_M V(x) :\phi^2_\varepsilon(x):dv\right) \right),\text{ when }d=(2,3) ,\\   
Z_g(\lambda,V)&=&\lim_{\varepsilon\rightarrow 0^+}\mathbb{E}\left(\exp\left(-\frac{\lambda}{2}\int_M V(x) :\phi^2_\varepsilon(x):dv - \frac{\lambda^2\int_MV^2(x)dv}{64\pi^2}\vert\log(\varepsilon) \vert \right) \right),\text{ when } d=4.
\end{eqnarray*} 
Then the sequence $c_n(g,V)$ is well--defined 
for $n>\frac{d}{2}$ and the partition functions $Z_g$
satisfies the following 
identity for small $\vert\lambda\vert$~:
\begin{equation}
Z_g(\lambda,V)=\exp\left( P(\lambda)+\sum_{n>\frac{d}{2}} \frac{(-1)^{n}c_n(g,V)\lambda^n}{2n}\right)
\end{equation}
where $P=c\lambda^2$ when $d=4$, $P=0$ when $d<4$ and $Z_g^{-2}$ extends as 
an \textbf{entire function} on the complex plane $\mathbb{C}$
whose zeroes lie in $-\sigma(V\Delta^{-1})$.
\end{prop}
Note that $V\Delta^{-1}$ 
is a pseudodifferential operator of negative degree
hence a compact operator 
and $\sigma(\Delta^{-1}V)$ is well--defined.
We observe that for $d=4$, the formal integral $c_2(g,V)$ is ill--defined. If it were well--defined, it would be understood 
as the limit
when $\varepsilon\rightarrow 0^+$: $\lim_{\varepsilon\rightarrow 0^+}\frac{1}{4}\mathbb{E}\left( :\phi^2(V): :\phi^2(V): \right)$ which diverges logarithmically. This divergence 
is subtracted by the counterterm 
$\frac{\int_MV^2(x)dv}{64\pi^2}\vert\log(\varepsilon) \vert$. The resulting finite part \textbf{is hidden in the constant} $c$
in the polynomial term $P(\lambda)$ which \textbf{depends on the metric $g$ and the function $V$}. But in the special case where $V=1$, we will see that $c$ depends on the metric $g$ \textbf{only through the spectrum} of $\Delta$.

At this point, it was pointed out to the author by Claudio Dappiaggi that there
should be some explicit relation between
the renormalization done here
and the methods from the papers~\cite{DappEucl,Ricci,DangHersc}
on Euclidean algebraic Quantum Field Theory which use an Euclidean version 
of Epstein--Glaser renormalization.
From the above, we deduce the following 
corollaries when $V=1\in C^\infty(M)$~:
\begin{coro}\label{c:coromain}
Let $(M_1,g_1),(M_2,g_2)$ be a pair of compact Riemannian manifolds without boundary of dimension $2\leqslant d\leqslant 4$, then 
the following claims are equivalent~: 
\begin{enumerate}
\item $c_n(g_1)=c_n(g_2)$ for all $n>\frac{d}{2}$, 
\item the partition functions coincide $Z_{g_1}=Z_{g_2}$
,
\item $(M_1,g_1),(M_2,g_2)$ are \textbf{isospectral}. 
\end{enumerate}

In particular the Einstein--Hilbert action $S_{EH}$, hence
the Euler characteristic $\chi(M)$ when $M$ is a surface, can be recovered from $Z_g$
by the formula~:
\begin{eqnarray*}
\boxed{ S_{EH}(g)=Res|_{s=\frac{d}{2}-1} \sum_{\lambda, Z_g(\lambda)^{-2}=0} \lambda^{-s} .   }
\end{eqnarray*}
and if $(g_1,g_2)$ are metrics with 
negative sectional curvatures s.t. $Z_{g_1}=Z_{g_2}$, 
then
$\mathcal{P}(g_1)=\mathcal{P}(g_2)$ where the length spectrum from definition~\ref{d:periods}
coincides with the singular support of the distribution~:
\begin{eqnarray}
\boxed{t\mapsto Re\left(\sum_{\lambda, Z_g(\lambda)^{-2}=0} e^{it\sqrt{\lambda}} \right) \in \mathcal{D}^\prime(\mathbb{R}_{>0}).}
\end{eqnarray}
\end{coro}

Real valued random variables $X$ are entirely characterised by their probability distribution
or equivalently by their generating function which is the formal Fourier--Laplace transform
of the probability distribution. 
In dimension $d=(2,3)$,
the main Theorem 
of our note deals with the 
rigidity of the 
Riemannian structure in negative curvature 
where the fluctuations of the
Wick square are encoded by the probability distribution of
the random variable $ \int_M:\phi^2(x):dv$ or the partition function $Z_g$ which should be understood as some kind of
Fourier--Laplace transform of the probability distribution of $ \int_M:\phi^2(x):dv$.
For $d=4$, $ \int_M:\phi^2(x):dv$ is no longer a random variable. Only the renormalized partition function $Z_g$ is well--defined and we obtain 
a similar rigidity result fixing the renormalized partition function.
Our Theorem is stated 
in terms of finite dimensional submanifolds $N$
of the regular part $\mathcal{G}$ of the 
moduli space of metrics. 
We will explain 
right after the statement of the Theorem
some 
subtle aspects of our assumptions.
\begin{thm}\label{t:mainthm1}
For every compact Riemannian manifold $(M,g)$ of dimension $2\leqslant d\leqslant 4$, $\phi$ is the Gaussian free field with covariance 
$\mathbf{G}$ with corresponding measure $\mu$.
Denote by $\phi_\varepsilon=e^{-\varepsilon\Delta}\phi$ to be the
heat regularized GFF.

If $d=2,3$
then the limit $ \int_M:\phi^2(x):dv=\lim_{\varepsilon\rightarrow 0^+} \int_M \phi_\varepsilon^2(x)dv-\mathbb{E}\left( \int_M \phi_\varepsilon^2(x)dv \right)  $
converges as a random variable in $L^p(\mathcal{D}^\prime(M),\mu), 2\leqslant p <+\infty$ with the following properties~:
\begin{enumerate}
\item Let $N$ be a \textbf{finite dimensional} submanifold of $ \mathcal{G}\subset \mathcal{R}\left( M\right)$
s.t. its boundary $\partial N$ is contained in the boundary $\partial\mathcal{G}$ of the regular part $\mathcal{G}$.
For all $\varepsilon>0$, the set of classes of metrics $[g]\in N\cap \mathcal{R}\left( M\right)_{\leqslant-\varepsilon}\cap \mathcal{G}_{\geqslant\varepsilon}$ such that the random variable $ \int_M:\phi^2(x):dv$
\textbf{has given probability distribution is finite}.

\item When $d=3$, 
for a sequence $(M_i,g_i)_{i\in \mathbb{N}}$ of Riemannian $3$--manifolds of negative curvature such that the random variable $ \int_{M_i}:\phi^2(x):dv_i$
\textbf{has a fixed given probability distribution}, 
the number of diffeomorphisms types represented in the sequence $(M_i)_i$ is \textbf{finite} and
one can extract a subsequence
such that $M_i$ has \textbf{fixed diffeomorphism type} and $g_i\rightarrow g$ for some metric $g$ in the $C^\infty$ 
topology.
\end{enumerate}

If $d=4$ then 
\begin{enumerate}
\item Let $N$ be a \textbf{finite dimensional} submanifold of $ \mathcal{G}\subset \mathcal{R}\left( M\right)$
s.t. its boundary $\partial N$ is contained in the boundary $\partial\mathcal{G}$ of the regular part $\mathcal{G}$.
For all $\varepsilon\in (0,1)$, the set of classes of metrics $[g]\in N\cap \mathcal{R}\left( M\right)_{[-\varepsilon^{-1},-\varepsilon]}\cap \mathcal{G}_{\geqslant\varepsilon}$ 
\textbf{with given partition function $Z_g$ is finite}.

\item 
For a sequence $(M_i,g_i)_{i\in \mathbb{N}}$ of Riemannian $4$--manifolds of negative sectional curvatures bounded 
in some compact interval such that the 
\textbf{partition function $Z_g$ is given}, 
the number of diffeomorphisms types represented in the sequence $(M_i)_i$ is \textbf{finite} and
one can extract a subsequence
such that $M_i$ has \textbf{fixed diffeomorphism type} and $(M_i,g_i)\rightarrow (M,g)$ in the \textbf{Lipschitz topology}.
\end{enumerate}
\end{thm}

Our result gives an example of metric dependent (non topological) 
Quantum Field Theory where the knowledge of 
the partition function gives both some \textbf{topological and metrical} constraints on the
Riemannian manifold $(M,g)$.

Let us comment on the definition of $N$. The set $N$ is a submanifold of the regular part $\mathcal{G}$ since only $\mathcal{G}$ has a manifold structure. The boundary $\partial N$ of $N$ is defined by taking the
closure of $N$ in $\mathcal{R}(M)$ for the topology of $\mathcal{R}(M)$, then 
$\partial N=\overline{N}\setminus \text{Int}\left(N\right)$ is considered as a subset of $\mathcal{R}(M)$. 
A subtle but important observation is that $N$ is not \textbf{necessarily compact}. Let us explain why and then give some example.
Both $\mathcal{R}(M)$ and $\mathcal{G}$, endowed with the 
Ebin metric $\textbf{d}$, have finite diameter. But the point is that bounded subsets for the metric $\textbf{d}$ are not necessarily bounded for the 
induced topology on $\mathcal{R}(M)$ as illustrated in the following~:
\begin{ex}\label{ex:1}
Choose any metric $g\in \textbf{Met}(M)$ on $M$ whose isometry group is reduced to the identity element. Hence the corresponding class
$[g]$ belongs to the regular part $\mathcal{G}$. 
Observe that the subset $\{tg \text{ s.t. } t>0\}\subset \textbf{Met}(M)$ is not bounded in the $C^\infty(M)$ topology~\footnote{since it is not even bounded for the $C^0$ norm.}. The metrics
$t_1g$ and $t_2g$ are not isometric if $t_1\neq t_2$ since they give different volumes for $M$, 
therefore each $tg$ gives a different class $[tg]\in \mathcal{G}$ and by quotient this defines
a non trivial subset $N=\{[tg] \text{ s.t. } t>0\}\subset \mathcal{G}$.
By definition of the quotient topology, the subset $N$ is not bounded in $\mathcal{G}$ for the
topology of $\mathcal{R}(M)$.  
\end{ex}
The next example provides a simple analogy
with the above phenomena.
\begin{ex}
Consider the Fr\'echet space
$C^\infty(\mathbb{S}^1)$ of smooth function on the circle $\mathbb{S}^1$. Then 
the smooth topology of $C^\infty\left(\mathbb{S}^1\right)$ is metrizable. Set $\Vert .\Vert_{H^s}$ to be the Sobolev norm of 
degree $s\in \mathbb{N}$ 
then consider the distance
$\textbf{d}(f,g)=\sum_{s=0}^\infty \frac{1}{2^s} \frac{\Vert f-g \Vert_{H^s}}{1+\Vert f-g \Vert_{H^s}} $  for $(f,g)\in C^\infty(\mathbb{S}^1)^2$.
Then by construction the diameter of $C^\infty(\mathbb{S}^1)$ for the distance $\textbf{d}$ equals $2$ but $C^\infty(\mathbb{S}^1)$ itself is not bounded.
\end{ex}
 
If we denote by $\iota:N\hookrightarrow \mathcal{G}$ the abstract embedding, then
the preimage of a bounded subset for $\textbf{d}$ is not necessarily bounded.
\begin{ex}
Consider again the $1$--dimensional manifold $N$ from example~\ref{ex:1}.
Then this defines an embedding $\iota:\mathbb{R}\simeq N\longmapsto \mathcal{G}$
and the preimage of $\mathcal{G}$ itself which is a bounded subset for 
$\textbf{d}$ is $\mathbb{R}$ which is not bounded. 
\end{ex} 
To overcome these difficulties we will make use of
compactness Theorems for isospectral metrics and also
the finite dimensionality of $N$ will play an important role 
in our proof.
\subsection{Acknowledgements.}

I would like to thank Thibault Lefeuvre, Marco Mazzucchelli and Colin Guillarmou for 
teaching me some methods from inverse problems which are used in the present paper and also
thanks to Claudio Dappiaggi, Michal Wrochna, 
Jan Derezi\'nski, 
Estanislao Herscovich 
and Christian G\'erard 
for keeping my interest and motivation 
for Quantum Field Theory on curved spaces. 
Finally, I would like to thank my wife Tho for creating the great atmosphere
that makes things possible.

\section{Proof of Proposition \ref{t:mainthm2}.}

The results of Proposition \ref{t:mainthm2} are particular cases
of the main results from~\cite{Dangquillen}.
However, since we are in low dimension $d\leqslant 4$ in the present case, we can give a simple, self--contained proof which relies on simple commutator arguments in pseudodifferential calculus and using the asymptotic expansion of the heat kernel.

\subsubsection{Quadratic perturbations of Gaussian measures.}
We recall the content of~\cite[Proposition 9.3.1 p.~211]{Glimm}, slightly adapted to our situation, which yields a relation between 
partition functions of small quadratic perturbations of some Gaussian field and some convergent power series.
We shall denote by $L^2(M)_0$ and $\mathcal{D}^\prime(M)_0$ the respective
closed subspaces of $L^2(M)$ and $\mathcal{D}^\prime(M)$ which are orthogonal to constants and
$\Vert A\Vert_{HS}:=\sqrt{Tr_{L^2}\left(A^*A \right)}$ denotes the Hilbert--Schmidt norm.
For any Hilbert space $H$, we denote by 
$\mathcal{B}\left(H,H\right)$ the algebra of bounded operators on $H$.
\begin{prop}\label{p:glimmjaffe}
Let $C$ be a bounded positive self--adjoint operator on $L^2(M)_0$ and $b$ real, symmetric s.t.
$ 0 < C^{-1}+b$ as quadratic forms. Denote by $d\mu_C$ the Gaussian measure on $\mathcal{D}^\prime(M)_0$
whose covariance is $C$.
Set $:\mathcal{V}:_C=\frac{1}{2} \int_{M\times M} :\phi(x) b(x,y)\phi(y):_C $
where $b(x,y)$ denotes the Schwartz kernel of $b$ and $:\phi(x) b(x,y)\phi(y):_C$ the Wick ordered operator w.r.t. the Gaussian measure $d\mu_C$.
If $\widehat{b}=C^{\frac{1}{2}}bC^{\frac{1}{2}}$ is Hilbert--Schmidt then both
$:\mathcal{V}:_C$ and $e^{-:\mathcal{V}:_C}$ are in $L^p(d\mu_C)$ for all $p<+\infty$ and
$$ \mathbb{E}\left(e^{-:\mathcal{V}:_C} \right)=\exp\left(-\frac{1}{2}Tr_{L^2}\left(\log(I+\widehat{b})-\widehat{b} \right)\right) $$
where the expansion in powers of $:\mathcal{V}:_C$ converges absolutely for $\Vert \widehat{b}\Vert_{HS}<1$.  
\end{prop}

In the sequel, we denote by $\Delta^{-1}$ the continuous linear map $\mathcal{D}^\prime(M)\mapsto \mathcal{D}^\prime\left(M\right)_0$ 
whose Schwartz kernel is the Green function $\mathbf{G}\in \mathcal{D}^\prime(M\times M)$ defined in definition~\ref{dfn:green}.
For all distribution $u\in \mathcal{D}^\prime(M)$,
$\Delta\left(\Delta^{-1}u\right)=\Delta^{-1}\left(\Delta u\right)=u-\frac{\int_Mu}{Vol(M)}$
which means $\Delta^{-1}$ acts as the inverse 
of $\Delta$ restricted to $\mathcal{D}^\prime(M)_0$.
In general in our paper, all powers $\Delta^{-s}, s\in \mathbb{R}$ of the Laplace operator $\Delta$ are defined using the spectral resolution as~:
$ \forall u\in \mathcal{D}^\prime(M), \Delta^{-s}u=\sum_{\lambda\in \sigma(\Delta)\setminus \{0\}} \lambda^{-s}\left\langle u,e_\lambda \right\rangle e_\lambda   $
where the r.h.s converges in $\mathcal{D}^\prime(M)$.

Set $\widehat{V_\varepsilon}=e^{-\varepsilon\Delta}\Delta^{-\frac{1}{2}}V\Delta^{-\frac{1}{2}}e^{-\varepsilon\Delta}$, then by definition of $\Delta^{-1}: L^2(M)\mapsto L^2(M)_0$, $\ker(\Delta^{-1})$ is reduced to the constant functions. Therefore we find that
for all $k\geqslant 1$,
\begin{equation}\label{e:remarkzeromodes}
Tr_{L^2_0}\left( \widehat{V_\varepsilon}^k \right)=Tr_{L^2}\left( \widehat{V_\varepsilon}^k \right).
\end{equation}
We will use the above identity
to switch between 
the two traces $Tr_{L^2}$ and $Tr_{L^2_0}$ when dealing with analytic functionals
of $\widehat{V_\varepsilon}$.
The above proposition \ref{p:glimmjaffe} applied to the covariance $C_\varepsilon=e^{-2\varepsilon\Delta}\Delta^{-1}$ and the quadratic perturbation $\frac{1}{2}\lambda\int_M V(x)\phi(x)^2dv$ together with identity \ref{e:remarkzeromodes} 
yields for
$ \vert \lambda\vert <\frac{1}{\Vert V \Vert_{L^\infty(M)}\Vert e^{-2\varepsilon\Delta}\Delta^{-1}\Vert_{HS}} $ ~:
\begin{eqnarray*}
\mathbb{E}\left( \exp\left(-\frac{\lambda}{2}\int_M V(x):\phi_\varepsilon^2(x):  dv\right)\right)&=&\exp\left(-\frac{1}{2} Tr_{L^2_0}\left(\log\left(I+\lambda \widehat{V_\varepsilon}  \right) - \lambda \widehat{V_\varepsilon}\right) \right)  \\
&=& \exp\left(-\frac{1}{2} Tr_{L^2}\left(\log\left(I+\lambda \widehat{V_\varepsilon}  \right) - \lambda \widehat{V_\varepsilon}\right) \right)
\end{eqnarray*}
where $\widehat{V_\varepsilon}$ is smoothing hence Hilbert--Schmidt and both series are absolutely convergent in $\lambda$ since  
\begin{eqnarray*}
\Vert e^{-\varepsilon\Delta}\Delta^{-\frac{1}{2}}V\Delta^{-\frac{1}{2}}e^{-\varepsilon\Delta} \Vert_{HS}^2
&=&\underset{\text{ by Lidskii since the operator is smoothing}}{
\underbrace{Tr_{L^2}\left(  e^{-\varepsilon\Delta}\Delta^{-\frac{1}{2}} V \Delta^{-1}e^{-2\varepsilon\Delta} Ve^{-\varepsilon\Delta}\Delta^{-\frac{1}{2}} \right)}}\\  
=\underset{\text{ by cyclicity of }Tr_{L^2}}{ \underbrace{ Tr_{L^2}\left(   V \Delta^{-1}e^{-2\varepsilon\Delta} V\Delta^{-1}e^{-2\varepsilon\Delta} \right)} }
&\leqslant& \underset{\text{ by Cauchy--Schwartz for }\Vert .\Vert_{HS}}{\underbrace{Tr_{L^2}\left( \Delta^{-1}e^{-2\varepsilon\Delta} 
 V^2 \Delta^{-1}e^{-2\varepsilon\Delta}  \right)}}\\
\leqslant \underset{\text{cyclicity again}}{\underbrace{Tr_{L^2}\left(  V^2 \Delta^{-2}e^{-4\varepsilon\Delta}  \right)}}
&\leqslant
&\underset{\text{H\"older}}{\underbrace{\Vert V\Vert_{\mathcal{B}(L^2,L^2)}^2Tr_{L^2}\left( \Delta^{-2}e^{-4\varepsilon\Delta}  \right)}}
\leqslant 
\Vert V\Vert_{L^\infty(M)}^2Tr_{L^2}\left( \Delta^{-2}e^{-4\varepsilon\Delta}  \right)
\end{eqnarray*}
where in the last inequality, we used the fact that $\Vert V\Vert_{L^\infty(M)}=\Vert V\Vert_{\mathcal{B}(L^2,L^2)}$. 

\subsubsection{Relating to functional determinants.}

Now we observe that expanding the $\log$ as a power series and
$Tr_{L^2}\left(\left(e^{-\varepsilon\Delta}\Delta^{-\frac{1}{2}}V\Delta^{-\frac{1}{2}}
e^{-\varepsilon\Delta}\right)^k \right) = Tr_{L^2}\left(\left(e^{-2\varepsilon\Delta}\Delta^{-1}V\right)^k \right)$
by cyclicity of the $L^2$--trace  
yields~:
\begin{eqnarray*}
&&\exp\left(-\frac{1}{2} Tr_{L^2}\left(\log\left(I+\lambda e^{-\varepsilon\Delta}\Delta^{-\frac{1}{2}}V\Delta^{-\frac{1}{2}}
e^{-\varepsilon\Delta} \right) - \lambda e^{-\varepsilon\Delta}\Delta^{-\frac{1}{2}}V\Delta^{-\frac{1}{2}}
e^{-\varepsilon\Delta}\right) \right) \\
& =& \exp\left(\frac{1}{2}\sum_{k=2}^\infty \frac{(-1)^k\lambda^k}{k}Tr_{L^2}\left(\left(e^{-\varepsilon\Delta}\Delta^{-\frac{1}{2}}V\Delta^{-\frac{1}{2}}
e^{-\varepsilon\Delta}\right)^k \right) \right)=
\exp\left(\frac{1}{2}\sum_{k=2}^\infty \frac{(-1)^k\lambda^k}{k}Tr_{L^2}\left(\left(e^{-2\varepsilon\Delta}\Delta^{-1}V\right)^k \right) \right)\\
&=&  \exp\left(-\frac{1}{2} Tr_{L^2}\left(\log\left(I+\lambda e^{-2\varepsilon\Delta}\Delta^{-1}V  \right) - \lambda e^{-2\varepsilon\Delta}\Delta^{-1}V\right) \right).
\end{eqnarray*}
Then by Lemma~\ref{l:detregtraces}, there is an explicit relation
connecting Fredholm determinants $\det_F$, Gohberg--Krein's determinants
$\det_2$ and functional traces (see also~\cite[p.~212]{Glimm}), 
this relation reads~:
\begin{eqnarray*}
&&\exp\left(-\frac{1}{2} Tr_{L^2}\left(\log\left(I+\lambda e^{-2\varepsilon\Delta}\Delta^{-1}V  \right) - \lambda e^{-2\varepsilon\Delta}\Delta^{-1}V\right) \right)\\
&=&\text{det}_F\left(I+\lambda e^{-2\varepsilon\Delta}\Delta^{-1}V \right)^{-\frac{1}{2}}\exp(\frac{\lambda}{2}Tr_{L^2}(e^{-2\varepsilon\Delta}\Delta^{-1}V))
=\text{det}_2\left(I+\lambda e^{-2\varepsilon\Delta}\Delta^{-1}V  \right)^{-\frac{1}{2}}
\end{eqnarray*}  
which follows immediately
from the properties of
Gohberg--Krein's determinants $\det_2$.

For the moment, for every $\varepsilon>0$ and $ \vert \lambda\vert <\frac{1}{\Vert V \Vert_{L^\infty(M)}\Vert e^{-2\varepsilon\Delta}\Delta^{-1}\Vert_{HS}} $, 
we obtained the relation
\begin{eqnarray}\label{e:partitiondeterminant}
\boxed{ \mathbb{E}\left( \exp\left(-\frac{\lambda}{2}\int_M V(x):\phi_\varepsilon^2(x):  dv\right)\right)=\text{det}_2\left(I+\lambda e^{-2\varepsilon\Delta}\Delta^{-1}V  \right)^{-\frac{1}{2}}}
\end{eqnarray}
relating the  
partition function of the regularized Wick square and the Gohberg--Krein determinant
for some regularized operator $ I+\lambda e^{-2\varepsilon\Delta}\Delta^{-1}V $ and where both sides can be expanded as convergent power series in $\lambda$ provided $ \vert \lambda\vert <\frac{1}{\Vert V \Vert_{L^\infty(M)}\Vert e^{-2\varepsilon\Delta}\Delta^{-1}\Vert_{HS}} $. 
For fixed $\varepsilon>0$, by analytic continuation property of Gohberg--Krein's determinant, both sides 
of equation~\ref{e:partitiondeterminant} extend as entire functions of $\lambda\in \mathbb{C}$. 

\subsubsection{The limit $\varepsilon\rightarrow 0^+$.}
The goal of this short paragraph is to study the
limit of the Fredholm operator $ I+\lambda e^{-2\varepsilon\Delta}\Delta^{-1}V $ when $\varepsilon\rightarrow 0^+$.
We will say that a pseudodifferential $A$ belongs to $\Psi^{+0}_{1,0}(M)$ if $A\in \Psi^s_{1,0}(M)$ 
for all $s>0$. 
\begin{lemm}[Microlocal convergence of heat operator]
Let $e^{-t\Delta}$ be the heat operator. Then we have the convergence
$ e^{-t\Delta}\underset{t\rightarrow 0^+}{\rightarrow} Id\text{ in }\Psi^{+0}_{1,0}(M)$. 
\end{lemm}
\begin{proof}
For every real number $s$,
a symbol
$p\in S^{s}_{1,0}\left(\mathbb{R} \right)$ iff $p$ is in $C^\infty\left( \mathbb{R}\right)$ and 
$ \vert \partial_\xi^jp(\xi) \vert\leqslant 
C_j\left(1+\vert \xi\vert\right)^{s-j} $~\cite[Lemm 1.2 p.~295]{Taylor-81} for every $j\in \mathbb{N}$.
Observe that the function $p_t:\xi\in \mathbb{R}\mapsto e^{-t\vert\xi\vert^2}$ 
defines a family $(p_t)_{t\in [0,+\infty)}$ of symbols in
$S^0_{1,0}\left(\mathbb{R} \right)$ such that $p_t\underset{t\rightarrow 0}{\rightarrow} 1$ 
in $S^{+0}_{1,0}\left(\mathbb{R} \right)$.
Indeed, for $k\in \mathbb{N}$ and
for $t$ in some compact interval $[0,a], a>0$, we find by direct computation that~:
$(1+\vert\xi\vert)^k\vert\partial_\xi^ke^{-t\xi^2} \vert\leqslant C(1+\vert\xi\vert)^k\sum_{0\leqslant l\leqslant \frac{k}{2}} t^{k-l}\vert\xi\vert^{k-2l} e^{-t\xi^2}$ where the constant $C$ depends
only on $k$.
 
 When $\vert\xi\vert\geqslant a$, the function
$t\in [0,+\infty)\mapsto (t^{k-l}\xi^{k-2l})e^{-t\xi^2}$ goes to $0$ when $t=0,t\rightarrow +\infty$ and 
reaches its maximum when $\frac{d}{dt}\left( (t^{k-l}\xi^{k-2l})e^{-t\xi^2}\right)
=((k-l)t^{k-l-1}\xi^{k-2l} - t^{k-l} \xi^{k-2l+2}) e^{-t\xi^2}
=((k-l)-t\xi^2) t^{k-l-1} \xi^{k-2l}e^{-t\xi^2}=0 $ for $t=\frac{k-l}{\xi^2}$.
Hence when $\vert\xi\vert\geqslant a$, $$\sup_{t\in [0,a]} (1+\vert\xi\vert)^k\vert(t^{k-l}\xi^{k-2l})\vert e^{-t\xi^2} \leqslant  (k-l)^{k-l}(1+\vert\xi\vert)^k\vert\xi\vert^{-k}\leqslant (k-l)^{k-l}(1+a^{-k})^k.$$
On the other hand, 
if $\vert \xi\vert\leqslant a$, $t\in [0,a]$, we find that
$(1+\vert\xi\vert)^k\vert\partial_\xi^ke^{-t\xi^2} \vert\leqslant  C(1+a)^k\sum_{0\leqslant l\leqslant \frac{k}{2}} a^{2k-3l}$.

Therefore, we showed that
$(1+\vert\xi\vert)^k\vert \partial^k_\xi e^{-t\xi^2} \vert\leqslant C_k $
uniformly on $t\in [0,a]$, hence $p_t\in S^0_{1,0}$ uniformly on $t\in [0,a]$.
We also have
for all $\delta,u>0$,
$t\leqslant\delta^{1+2u}$
implies that
$\sup_{\xi} \vert(1+\vert \xi\vert)^{-u} (e^{-t\xi^2}-1)\vert\leqslant \delta $ which means that 
$\sup_{\xi} \vert(1+\vert \xi\vert)^{-u} (e^{-t\xi^2}-1)\vert\rightarrow 0 $
when $t\rightarrow 0^+$ which implies the convergence $p_t\rightarrow 1$ in $S^{+0}_{1,0}$.
By a result of Strichartz~\cite[Thm 1.3 p.~296]{Taylor-81}, 
\begin{equation}\label{e:convheat}
p_t(\sqrt{\Delta})=e^{-t\Delta}\underset{t\rightarrow 0^+}{\rightarrow} Id\text{ in }\Psi^{+0}_{1,0}(M).
\end{equation}
\end{proof}

\subsubsection{No counterterms for $d\leqslant 3$.}
We now discuss the case of dimension $d=(2,3)$ where we show that the regularized partition function converges
when $\varepsilon\rightarrow 0^+$
and we do not need to subtract local counterterms.
By composition of pseudodifferential operators, we find that
$e^{-2\varepsilon\Delta}\Delta^{-1}V \underset{\varepsilon\rightarrow 0}{\rightarrow} \Delta^{-1}V$ in the space $\Psi^{-2+0}(M)$
of pseudodifferential operators of order $-2+\varepsilon, \forall \varepsilon >0$ which implies that
the convergence occurs in the ideal $\mathcal{I}_2$ of Hilbert--Schmidt operators by~\cite[Prop B 21]{DyZwscatt}.  
By continuity of Gohberg--Krein's determinant on the ideal $\mathcal{I}_2$ i.e. of the map 
$H\in \mathcal{I}_2\mapsto \det_2\left(I+H\right) $~\cite[Thm 9.2]{Simon-traceideals}, we find that
\begin{equation}
Z_g(\lambda)=\lim_{\varepsilon\rightarrow 0^+}  \mathbb{E}\left( \exp\left(-\frac{\lambda}{2}\int_M V(x):\phi_\varepsilon^2(x):  dv\right)\right)= 
\text{det}_2\left(I+\lambda \Delta^{-1} V \right)^{-\frac{1}{2}}.
\end{equation}

By Lemma~\ref{l:detregtraces}, the function $\lambda \mapsto\text{det}_2\left(I+\lambda \Delta^{-1}V  \right)$ has an analytic continuation to the complex plane as an entire function whose zeroes is exactly the set $\{\lambda\in \mathbb{C} \text{ s.t. }\ker\left(  I+\lambda \Delta^{-1}V\right)\neq \{0\} \}$. Thus, we  
find that the divisor of $Z_g^{-2}$ coincides with the subset $\{ \lambda\text{ s.t. }z\lambda=-1, z\in\sigma(\Delta^{-1}V) \}\subset \mathbb{C}$ hence when $V=1$, the partition function $Z_g$ determines 
the spectrum $\sigma(\Delta)$ of the Laplace--Beltrami operator $\Delta$.

\subsection{Explicit counterterms in dimension $d\leqslant 4$.}
%
In what follows, for a separable Hilbert space $H$ and every integer $p\geqslant 1$, 
we shall denote by $\mathcal{I}_p(H)$ the Schatten ideal of operators
whose $p$-th power is trace class. When there is no ambiguity on $H$, we will shortly
write $\mathcal{I}_p$. The space $\mathcal{I}_p(H)$ 
is endowed with the Schatten norm $\Vert.\Vert_{\mathcal{I}_p}$.
$\mathcal{I}_1(H),\mathcal{I}_2(H)$ are the usual ideals of trace class and Hilbert--Schmidt operators 
respectively.

When $d=(2,3)$, $\Delta^{-1}$ is only Hilbert--Schmidt but not trace class and we only need the Wick renormalization
to renormalize the partition function. This is exactly what 
Gohberg--Krein's renormalized determinant $\det_2$ is doing. 
When $d=4$, for $\vert \lambda\vert<\frac{1}{\Vert V\Vert_{L^\infty(M)} \Vert e^{-2\varepsilon\Delta}\Delta^{-1}\Vert_{HS}}$, we start again from the series expansion~:
\begin{eqnarray*}
&&\log \mathbb{E}\left( \exp\left(-\frac{\lambda}{2}\int_M V(x):\phi_\varepsilon^2(x):  dv\right)\right)
= \frac{\lambda^2}{4} Tr_{L^2}\left(\left(e^{-2\varepsilon\Delta}\Delta^{-1}V\right)^2\right)\\
&+&
\frac{1}{2}\sum_{k=3}^{\infty} \frac{(-1)^{k}\lambda^{k}}{k}Tr_{L^2}\left( \left(e^{-2\varepsilon\Delta}\Delta^{-1}V\right)^k \right)
\end{eqnarray*}
where we need to renormalize 
$Tr_{L^2}\left(\left(e^{-2\varepsilon\Delta}\Delta^{-1}V\right)^2\right)$ since for all $k\geqslant 3$, equation~\ref{e:convheat} implies $ \left(e^{-2\varepsilon\Delta}\Delta^{-1}V\right)^k\underset{\varepsilon\rightarrow 0^+}{\rightarrow} \left(\Delta^{-1}V\right)^k \in \Psi^{-2k}(M)\subset \Psi^{-6}(M)$ which are trace class. 
We shall use pseudodifferential calculus to extract the singular part of this term. 
The extraction of the singular part would be easy if we considered the term
$ Tr_{L^2}\left(e^{-4\varepsilon\Delta}\Delta^{-2}V^2 \right) $ using the asymptotic expansion of the heat kernel. But 
as usual, the difficulty lies in the fact that operators do not commute
hence $ \left(e^{-2\varepsilon\Delta}\Delta^{-1}V\right)^2\neq  e^{-4\varepsilon\Delta}\Delta^{-2}V^2 $. The trick is to
arrange the term  $ \left(e^{-2\varepsilon\Delta}\Delta^{-1}V\right)^2$ to produce a commutator term
which is trace class~:
\begin{eqnarray*}
Tr_{L^2}\left(\left(e^{-2\varepsilon\Delta}\Delta^{-1}V\right)^2\right)&=&Tr_{L^2}\left(e^{-4\varepsilon\Delta}\Delta^{-2} V^2\right)+
Tr_{L^2}\left(\underset{\in \Psi^{-5}\text{ hence trace class}}{\underbrace{e^{-2\varepsilon\Delta}\Delta^{-1} [V,e^{-2\varepsilon\Delta}\Delta^{-1}]V}}\right),
\end{eqnarray*}
where the family of heat operators $(e^{-\varepsilon\Delta})_{\varepsilon\in [0,1]}$ is bounded in $\Psi^0(M)$ by equation~\ref{e:convheat}, the commutator term
$[V,e^{-2\varepsilon\Delta}\Delta^{-1}]$ is therefore bounded in
$\Psi^{-3}(M)$ \textbf{uniformly} in
the parameter $\varepsilon\in [0,1]$~\cite[p.~14]{Taylor2}. 
By composition in the pseudodifferential calculus and properties of the commutator of pseudodifferential operators,
we thus find that
$
e^{-2\varepsilon\Delta}\Delta^{-1} [V,e^{-2\varepsilon\Delta}\Delta^{-1}]V \in \Psi^{-5}\left(M \right),
$
uniformly in $\varepsilon\in [0,1]$ 
and is therefore of trace class by Proposition~\cite[Prop B 21]{DyZwscatt} since we are in dimension $d=4$, 
uniformly in the small parameter $\varepsilon\in [0,1]$. 
So we found
that 
$Tr_{L^2}\left(\left(\Delta^{-1} e^{-\varepsilon\Delta} Ve^{-\varepsilon\Delta}\right)^2\right)=Tr_{L^2}\left(e^{-4\varepsilon\Delta}\Delta^{-2} V^2\right)+\mathcal{O}(1) $, the singular part of $Tr_{L^2}\left(\left(\Delta^{-1} e^{-\varepsilon\Delta} Ve^{-\varepsilon\Delta}\right)^2\right)$ 
coincides with that of $Tr_{L^2}\left(e^{-4\varepsilon\Delta}\Delta^{-2} V^2\right)$. 
Then the singular part of $Tr_{L^2}\left(e^{-4\varepsilon\Delta}\Delta^{-2} V^2\right)$ is easily extracted using
the heat kernel asymptotic expansion~\cite{BGV}. First by~\cite[Proposition 3.3 p.~12]{DangZhang} based on the functional calculus of
the Laplace operator and Mellin transform, we have 
$ \Delta^{-2}=\frac{1}{\Gamma(2)}\int_0^\infty (e^{-t\Delta}-\Pi)tdt $
where both sides are defined a priori as elements in $\mathcal{B}(L^2(M),L^2(M))$ 
and $\Pi$ is the orthogonal projector on $\ker(\Delta)$ i.e. constant functions.
Hence $e^{-4\varepsilon\Delta}\Delta^{-2} V^2=\frac{1}{\Gamma(2)}\int_0^\infty e^{-(t+4\varepsilon)\Delta}(Id-\Pi)V^2 tdt $. We use the notation $\mathcal{O}(1)$ to refer to something which is bounded when $\varepsilon\rightarrow 0^+$. First we have the decomposition~:
\begin{eqnarray*}
&&Tr_{L^2}\left(e^{-4\varepsilon\Delta}\Delta^{-2} V^2\right)=\int_0^1 Tr_{L^2}\left(e^{-(4\varepsilon+t)\Delta}V^2 \right) tdt\\
&+&
\underset{\mathcal{O}(1)}{\underbrace{\int_{0}^1 Tr_{L^2}\left(-\Pi V^2\right) tdt+\int_1^{\infty} Tr_{L^2}\left(e^{-(t+4\varepsilon)\Delta}(Id-\Pi)V^2\right) tdt}}
\end{eqnarray*}
since $e^{-(t+4\varepsilon)\Delta}\left(-\Pi\right)=-\Pi$ and 
the integral $\int_1^{\infty} Tr_{L^2}\left(e^{-(t+4\varepsilon)\Delta}(Id-\Pi)V^2\right) tdt$
converges uniformly in $\varepsilon\rightarrow 0$ by exponential decay in $t$ of $Tr_{L^2}\left(e^{-(t+4\varepsilon)\Delta}(Id-\Pi)V^2\right) $. The proof of the exponential decay follows~\cite[7.3 Proof of Lemma 4.1]{DangZhang} and is a consequence of the \textbf{spectral gap} 
for $ e^{-(t+4\varepsilon)\Delta}(Id-\Pi)$. 
Then we may use the asymptotic expansion of the heat kernel~\cite[Thm 2.30]{BGV} to study the term
$e^{-(t+4\varepsilon)\Delta}(x,x)=\frac{1}{(4\pi(t+4\varepsilon))^2}+\mathcal{O}((t+4\varepsilon)^{-1})$. 
This yields~:
\begin{eqnarray*}
&&Tr_{L^2}\left(e^{-4\varepsilon\Delta}\Delta^{-2} V^2\right)=\int_0^1 Tr_{L^2}\left(e^{-(4\varepsilon+t)\Delta}V^2 \right) t dt+\mathcal{O}(1)
\\&=&\frac{1}{(4\pi)^2}\int_0^1 \frac{t}{(4\varepsilon+t)^2}dt\int_M V^2(x)dv+\mathcal{O}(1)
=\frac{\int_M V^2(x)dv}{16\pi^2} \int_{4\varepsilon}^{1+4\varepsilon}  (u^{-1} -4\varepsilon u^{-2}) du +\mathcal{O}(1)\\
&=&\frac{-\log(\varepsilon)\int_M V^2(x)dv}{16\pi^2}+\mathcal{O}(1)
= \frac{\int_M V^2(x)dv}{16\pi^2}\vert\log(\varepsilon)\vert +\mathcal{O}(1).
\end{eqnarray*}  

We conclude by the observation that for $\vert \lambda\vert<\frac{1}{\Vert V\Vert_{L^\infty(M)}\Vert \Delta^{-1}\Vert_{\mathcal{I}_3}}$, $Z_g(\lambda)$
\begin{eqnarray*}
&=&\lim_{\varepsilon\rightarrow 0^+}\mathbb{E}\left(\exp\left(-\frac{\lambda}{2}\int_M V(x) :\phi^2_\varepsilon(x):dv - \frac{\lambda^2\int_MV^2(x)dv}{64\pi^2}\vert\log(\varepsilon) \vert \right) \right)\\
&=&\lim_{\varepsilon\rightarrow 0^+}\exp\left(\underset{\mathcal{O}(1)}{\underbrace{ \frac{\lambda^2}{4}Tr_{L^2}\left(\left(\Delta^{-1} e^{-2\varepsilon\Delta} V\right)^2 \right)- \frac{\lambda^2\int_MV^2(x)dv}{64\pi^2}\vert\log(\varepsilon) \vert}} +
\sum_{k=3}^{\infty} \frac{(-1)^{k}\lambda^{k}}{2k}Tr_{L^2}\left( \left(e^{-2\varepsilon\Delta}\Delta^{-1}V\right)^k \right)\right)\\
&=&\lim_{\varepsilon\rightarrow 0^+} \exp\left( \frac{\lambda^2}{4}Tr_{L^2}\left(\left(\Delta^{-1} e^{-2\varepsilon\Delta} V\right)^2 \right)- \frac{\lambda^2\int_MV^2(x)dv}{64\pi^2}\vert\log(\varepsilon) \vert \right)\text{det}_3\left(I+\lambda\Delta^{-1}e^{-2\varepsilon\Delta} V\right)^{-\frac{1}{2}}\\
&=&e^{P(\lambda)}\text{det}_3\left(I+\lambda\Delta^{-1} V \right)^{-\frac{1}{2}}
\end{eqnarray*}
where we recognized Gohberg--Krein's renormalized determinant $\det_3$. By the properties of $\det_3$ recalled in Lemma~\ref{l:detregtraces}, the expression on the r.h.s. converges when expanded as power series in $\lambda$ provided $\vert \lambda\vert<\frac{1}{\Vert V\Vert_{L^\infty(M)}\Vert \Delta^{-1}\Vert_{\mathcal{I}_3}} $ 
since $\Delta^{-1}e^{-2\varepsilon\Delta} V\underset{\varepsilon\rightarrow 0^+}{\rightarrow} \Delta^{-1}V\in \Psi^{-2}(M)$ hence in the Schatten ideal $\mathcal{I}_3$ and $P$ is a polynomial of degree $2$.

Now we conclude similarly as for dimension $d=(2,3)$, by Lemma~\ref{l:detregtraces}, $\det_3\left(I+\lambda\Delta^{-1} V \right)$
has analytic continuation as an entire function in $\lambda\in\mathbb{C}$ which
vanishes with multiplicity on the set $-\sigma(\Delta^{-1}V)$
which implies that $Z_g$ determines $\sigma(\Delta)$ when $V=1$.

\subsubsection{Conclusion of the proof.}

The proof of identity
\begin{equation}
Z_g(\lambda,V)=\exp\left( P(\lambda)+\sum_{n>\frac{d}{2}} \frac{(-1)^{n}c_n(g,V)\lambda^n}{2n}\right)
\end{equation}
follows immediately from the fact that
for $n>\frac{d}{2}$, composition in the pseudodifferential calculus implies that 
$(\Delta^{-1}V)^n\in\Psi^{-2n}(M)$ is trace class hence
the integrals 
$$c_n(g,V)=\int_{M^n} t_n(x_1,\dots,x_n)V(x_1)\dots V(x_n)dv_n $$ are convergent and equal to
$Tr_{L^2}\left((\Delta^{-1}V)^n \right)$ and $P$ vanishes if the dimension 
$d\leqslant 3$. The conclusion follows from the relation 
of Gohberg--Krein's determinants $\det_p$ with functional traces summarized in Lemma~\ref{l:detregtraces}.

\section{Proof of Corollary \ref{c:coromain}.}

Let us prove the equivalence of claims 1),2),3) in
Corollary \ref{c:coromain}. In this paragraph, we shall denote the Laplace-Beltrami operator of the respective metrics
$g_i, i=1,2$ by $\Delta_{g_i}, i=1,2$ to stress the dependence in the metric.

\begin{itemize}
\item Let us first show that 1)$\implies$ 3) namely if some infinite number 
of Feynman amplitudes coincide then the metrics are isospectral. Set 
$[\frac{d}{2}]=\sup_{k\leqslant \frac{d}{2},k\in \mathbb{Z}}k$.
Observe that arguing as in the proof of Proposition~\ref{t:mainthm1}, when $\vert\lambda\vert<\frac{1}{\Vert \Delta^{-1}\Vert_{[\frac{d}{2}]+1} }$, the series
$ \exp\left( \sum_{n>\frac{d}{2}} \frac{(-1)^{n+1}c_n(g_i)\lambda^n}{n}\right), i=1,2$ converges absolutely to
the Gohberg--Krein determinant
$\det_{[\frac{d}{2}]+1}\left( Id+\lambda\Delta_{g_i}^{-1} \right) $.
So the coincidence of Feynman amplitudes $c_n(g_1)=c_n(g_2), \forall n>\frac{d}{2}$ implies 
the equality $\det_{[\frac{d}{2}]+1}\left( Id+\lambda\Delta_{g_1}^{-1} \right)=\det_{[\frac{d}{2}]+1}\left( Id+\lambda\Delta_{g_2}^{-1} \right) $
as entire functions by analytic continuation. Hence $(g_1,g_2)$ are isospectral by the properties of the zeros of $\det_{[\frac{d}{2}]+1}$.

\item Assume $3)$ namely that
$(M_1,g_1)$ and $(M_2,g_2)$ are isospectral. Our goal is to show how to recover the 
partition function $Z_g$ from the spectrum. 
In proposition~\ref{t:mainthm1}, we established the relation $Z_g(\lambda)=\det_2\left(Id+\lambda\Delta^{-1} \right)^{-\frac{1}{2}}$ when $d\leqslant 3, \vert\lambda\vert< \frac{1}{\Vert \Delta^{-1}\Vert_{HS}}$
and $Z_g(\lambda)=\lim_{\varepsilon\rightarrow 0^+}\exp\left(-\frac{\lambda^2\text{Vol}_g(M)}{64\pi^2}\vert\log(\varepsilon)\vert\right)\det_2\left(Id+\lambda e^{-2\varepsilon\Delta}\Delta^{-1} \right)^{-\frac{1}{2}}=e^{P(\lambda)}\det_3\left(Id+\lambda \Delta^{-1} \right)^{-\frac{1}{2}}$ when $d=4, \vert\lambda\vert< \frac{1}{\Vert \Delta^{-1}\Vert_{3}}$ and where $P$ is some polynomial of degree $2$. The r.h.s of both equalities are \textbf{purely spectral} since~: 
\begin{enumerate}
\item by Lemma~\ref{l:detregtraces}, the Gohberg--Krein determinants can be expressed in terms of 
$Tr_{L^2}\left( \Delta^{-n} \right)=\sum_{\lambda\in \sigma(\Delta)\setminus \{0\}}\lambda^{-n}=c_n(g), \forall n> \frac{d}{2}$ 
by Lidskii's Theorem,
\item for the $d=4$ case, we use the fact that $\text{Vol}_g(M)$ is spectral~\footnote{which follows from Weyl's law.} and $Tr_{L^2}\left( \left(e^{-2\varepsilon\Delta}\Delta\right)^{-n} \right)=\sum_{\lambda\in \sigma(\Delta)\setminus \{0\}} e^{-2n\varepsilon\lambda}\lambda^{-n}$, 
\end{enumerate}
which imply both
3)$\implies$2) and 3)$\implies$ 1). 

\item Finally, assume $2)$ namely that the partition functions are equal $Z_{g_1}=Z_{g_2}$. The equality
implies $\det_{[\frac{d}{2}]+1}\left( Id+\lambda\Delta_{g_1}^{-1} \right)=\det_{[\frac{d}{2}]+1}\left( Id+\lambda\Delta_{g_2}^{-1} \right) $ by the formula relating the partition functions and Gohberg--Krein's determinants hence 
$(g_1,g_2)$ are isospectral and 2) $\implies$ 1) again by $c_n(g)=Tr_{L^2}\left(\Delta^{-n} \right)=\sum_{\lambda\in \sigma(\Delta)\setminus \{0\}}\lambda^{-n}$.
\end{itemize}
We proved the desired equivalences.

Now the main implication of Corollary~\ref{c:coromain} 
is a consequence of
the deep Theorem of Colin de Verdi\`ere~\cite{CDV1,CDV2}, Duistermaat--Guillemin~\cite[Thm 4.5 p.~60]{DG75}
relating the spectrum of the Laplacian and the length spectrum in negative curvature.   
We recall, in the particular case
of metrics with negative curvature~: 
\begin{thm}[Trace formula]\label{t:DGtrace}
Let $(M,g)$ be a smooth 
compact Riemannian manifold with 
negative sectional curvatures and $\Delta$ the Laplace-Beltrami operator.
Then 
the spectrum $\sigma\left(\Delta\right)$
determines
the non marked length spectrum
by the \textbf{trace formula}~:
\begin{equation}
2Re\left(\sum_{\lambda\in \sigma\left(\Delta \right)} e^{i\sqrt{\lambda}t} \right)=\sum_{\gamma} \frac{\ell_\gamma }{m_{\gamma}\vert \det\left(I-P_\gamma \right) \vert^{\frac{1}{2}}} \delta\left(t-\ell_\gamma\right) + L^1_{\text{loc}} , 
\end{equation} 
\footnote{which is an equality in the sense of distributions
in $\mathcal{D}^\prime\left(\mathbb{R}_{>0}\right)$} where $\ell_\gamma$, $m_\gamma$ are the period and multiplicity of the 
orbit $\gamma$ and $P_\gamma$ is the Poincar\'e return map.
Furthermore, the singularities of the wave trace equals the 
length spectrum~:
\begin{equation}
\text{singular support}\left(2Re\left( \sum_{\lambda\in \sigma\left(\Delta \right)} e^{i\sqrt{\lambda}t}\right)\right)=\{\ell_\gamma | [\gamma]\in \pi_1\left(M\right)  \}
\end{equation}
which implies the Laplace 
spectrum $\sigma\left(\Delta\right)$ determines the length spectrum of $\left(M,g\right)$. 
\end{thm}
For geodesic flows in negative curvature,  
the set of periods forms a discrete subset of $\mathbb{R}_{>0}$ hence each period is isolated
and the corresponding periodic orbits are isolated and in finite number. 
In that case, \cite[Theorem 3 p.~495]{guillemin1977lectures}) gives a leading term for the real part of the distributional flat trace
$2Re\left( Tr^\flat\left(U(.) \right)\right)\in \mathcal{D}^\prime\left(\mathbb{R}_{>0} \right)$ of the wave propagator $U(t)=e^{it\sqrt{\Delta}}$
of the form~:
\begin{eqnarray*}
2Re\left( Tr^\flat\left(U(t) \right)\right)=\sum_{[\gamma]\in \pi_1\left(M\right)} \frac{i^{-\sigma_\gamma}\ell_\gamma }{ m_{\gamma}\vert \det\left(I-P_\gamma \right) \vert^{\frac{1}{2}}} \delta\left(t-\ell_\gamma\right) + L^1_{\text{loc}}.
\end{eqnarray*}
The flat trace $Tr^\flat(U(t))$ of the wave propagator $U(t)$, also called 
distributional trace of $U(t)$, is defined to be the integral of the Schwartz kernel of $U(t)$ on the diagonal: $t\mapsto Tr^\flat(U(t))=\int_M K_t(x,x)dv(x)$ and it is a \textbf{distribution in the variable} $t$.
This formula holds true for every geodesic flow 
whose periodic orbits are countable (form a discrete set)
and such that each periodic orbit is \textbf{nondegenerate} in the sense the Poincar\'e map $P_\gamma$ is hyperbolic. 
In case the metric has negative curvature, each closed geodesic make a 
non-zero contribution to the singular support of $U$ since the Maslov index 
$\sigma_\gamma=0 $ for all $\gamma$ as noted in~\cite[Coro 1.1 p.~73]{BPP}.

The identity 
\begin{eqnarray*}
S_{EH}(g)=Res|_{s=\frac{d}{2}-1} \sum_{\lambda, Z_g(\lambda)^{-2}=0} \lambda^{-s}
\end{eqnarray*}
follows immediately from the spectral interpretation of the integral of the scalar curvature (Einstein--Hilbert action)~\cite[Thm 6.1 p.~26]{ammann2002einstein}. Let us briefly recall the principle of 
this derivation. The first heat invariant 
of the scalar Laplacian 
is directly related to the scalar curvature, 
for $Re(s)>\frac{d}{2}$, the sum
$\sum_{\lambda\in \sigma(\Delta),\lambda>0} \lambda^{-s}$ converges by Weyl's law
and coincides with $Tr_{L^2}\left(\Delta^{-s} \right)$. By the heat kernel expansion, the
trace $Tr_{L^2}\left(\Delta^{-s} \right)$ admits an analytic continuation
as a meromorphic function whose poles at $s=\frac{d}{2}-1$ are related to the first heat invariant.  

\section{Proof of Theorem \ref{t:mainthm1}.}

Let us explain the central ideas in the proof of Theorem \ref{t:mainthm1}.
Recall that we denoted by $N$ some smooth finite dimensional submanifold
of $\mathcal{G}\subset \mathcal{R}(M)$ such that $\partial N\subset \partial \mathcal{G}$.
When $d\leqslant 3$,
our goal is to show that given such $N$, 
the probability distribution of the random variable $ \int_M :\phi^2(x):dv$ determines
a \textbf{finite number} of elements in $N\cap \mathcal{R}\left( M\right)_{\leqslant-\varepsilon}\cap \mathcal{G}_{\geqslant\varepsilon}$.   
First,
the probability distribution of the random variable $ \int_M :\phi^2(x):dv$ determines the moments
$\mathbb{E}\left( \left(\int_M :\phi^2(x):dv\right)^k \right), k\geqslant 2$ of $ \int_M :\phi^2(x):dv$, 
hence the partition
function $Z_g(\lambda)$ whose zeroes give the spectrum $\sigma(\Delta)$ of the Laplace operator by Proposition~\ref{t:mainthm1}. 
Then the length spectrum $\mathcal{P}(g)$ of $M$ is recovered from the Laplace spectrum using 
the trace formula of Duistermaat--Guillemin.
This is the content of subsubsections~\ref{ss:existencewick} and~\ref{ss:speclength}.
When $d=4$, the discussion is simpler since 
we are directly given the partition function $Z_g$ which determines the Laplace spectrum 
by Proposition~\ref{t:mainthm1}.

Now it remains to show that in some 
finite dimensional submanifold $N$ in $ \mathcal{R}\left( M\right)_{\leqslant-\varepsilon}\cap \mathcal{G}_{\geqslant\varepsilon}$ for $d\leqslant 3$ and
$ \mathcal{R}\left( M\right)_{[-\varepsilon^{-1},-\varepsilon]}\cap \mathcal{G}_{>\varepsilon}$ when $d=4$, 
the Laplace spectrum together with the length spectrum
determine a \textbf{finite number} of isometry types. This follows from Proposition~\ref{p:rigidity} and 
let us explain informally the intuition behind this result.
Near every isometry class $[g_0]$ in $N$, one should think that 
there is some neighborhood $U\subset N$ of $[g_0]$ such that 
the map
$$ \text{Isometry class of metric }[g]\in U\subset\tilde{N} \longmapsto \text{length spectrum } \mathcal{P}([g])\subset \mathbb{R}_{>0} $$
is \textbf{injective}. In fact in our proof, we do not deal directly with this nonlinear map (nonlinear in the metric $g$) but with its linearization which is the
X-ray transform.

 Thinking in terms of representatives instead of isometry class, it means that
in some neighborhood of every metric $g_0\in N$, two metrics $(g_1,g_2)\in N^2$ with the same length spectrum must be isometric.
Then the finiteness follows from the compactness properties of isospectral metrics and finite dimensionality of $N$. 
Note that for simplicity of exposition, we prove Proposition~\ref{p:rigidity} by a
contradiction argument but the reader should keep in mind the intuitive picture explained above.

\subsubsection{Existence of Wick square as random variable.}
\label{ss:existencewick}
%
%
We use Proposition~\ref{p:glimmjaffe} on Gaussian measures. 
In dimension $d=(2,3)$, 
the operator $\Delta^{-\frac{1}{2}}V\Delta^{-\frac{1}{2}}\in \Psi^{-2}(M)$ is 
Hilbert--Schmidt and therefore the Wick renormalized 
functional
$ \int_M V:\phi^2(x):dv  $ is a \textbf{well--defined random variable}
in all $L^p(\mathcal{D}^\prime(M),\mu), p\in [2,+\infty)$ 
where $\mu$ is the Gaussian Free Field 
measure on
$\mathcal{D}^\prime(M)$ 
with covariance $\Delta^{-1}$.
From now on, we let $V=1\in C^\infty(M)$. 

\subsubsection{Spectrum of $\Delta$ and probability distribution of the Wick square $\int_M :\phi^2(x):dv$.}
\label{ss:speclength}
Furthermore, the 
probability distribution of the random variable $ \int_M :\phi^2(x):dv$, more precisely its moments are 
related to the partition
function $Z_g(\lambda)$ by the observation that
the series
\begin{eqnarray*}
Z_g(\lambda)=\sum_{n=0}^\infty \frac{(-1)^n\lambda^n}{2^nn!}\mathbb{E}\left(  \left(\int_M :\phi^2(x):dv\right)^n \right)
\end{eqnarray*} 
converges absolutely for 
$\vert\lambda \vert<\frac{1}{\Vert \Delta^{-1}\Vert_{HS} }$ where $\Vert.\Vert_{HS}$ is the Hilbert--Schmidt norm.
Therefore by Proposition~\ref{t:mainthm2}, the 
\textbf{probability distribution of $ \int_M :\phi^2(x):dv$ determines $Z_g$ and its zeroes, hence the spectrum of
$\Delta$}. 

To prove the second claim of Theorem~\ref{t:mainthm2}, we need to recall some
results on compactness of isospectral metrics which will also be useful for the proof of 
the first claim of Theorem~\ref{t:mainthm2}.
\subsubsection{Compactness of isospectral metrics.}
\label{ss:compactness}
A main ingredient of our proof of Theorem~\ref{t:mainthm2}
is compactness of the space of isospectral metrics.  
Note that two isospectral Riemannian surfaces
$(M_1,g_1)$ and $(M_2,g_2)$ have the same genus
since the second heat invariant $a_1=\frac{1}{6} \int_M \mathfrak{K}_g$ is a spectral invariant and is proportional to the integral of the scalar curvature $\mathfrak{K}_g$ on $M$. Therefore, $a_1$
determines the Euler characteristic thus the genus of $M$ by Gauss--Bonnet.
We start by the compactness result of Osgood--Philips--Sarnak~\cite{OPS} which deals with isospectral families surfaces.
\begin{thm}[Compactness for $d=2$]\label{t:compactd2}
An
isospectral set of isometry classes of metrics on a closed surface is sequentially compact in the $C^\infty$--topology.
\end{thm}

A general compactness result for isospectral metrics is no longer true in dimension $d=3$, we need some 
further assumptions on the metric.
For $d=3$, we shall use the celebrated result of Brooks-Petersen--Perry~\cite[Thm 0.2 p.~68]{BPP} and Anderson~\cite[Thm 0.1 p.~700]{andersonisospec}.
\begin{thm}[Compactness for $d=3$]\label{t:compactd3}
The space of smooth compact isospectral $3$-manifolds $(M,g)$
for which the length of the shortest closed geodesic is bounded from below
\begin{equation}\label{e:shortgeod}
\ell_M\geqslant \ell >0
\end{equation}
is compact in the $C^\infty$ topology. It follows that there are only finitely many
diffeomorphism types of isospectral $3$-manifolds which satisfy \ref{e:shortgeod}.
\end{thm}
Let us explain the meaning of the above 
statement in practice. Let $(M_i,g_i)_i$ denotes a sequence of isospectral smooth compact 
$3$--manifolds without boundary whose shortest closed geodesic 
has length bounded from below. Then there is a finite number of manifolds
$\left(M^\prime_1,\dots,M^\prime_k\right)$ and on each $M^\prime_j$ a compact family of metrics
$\mathcal{M}^\prime_j$ such that each of the manifolds $M_j$ is diffeomorphic to one of the $M^\prime_i$ and isometric to an element
of $\mathcal{M}^\prime_i$.

In dimension $d=4$, we use the following Theorem by Zhou~\cite[Thm 1.1 p.~188]{Zhou} which requires an additional assumption on 
sectional curvatures
~: 
\begin{thm}[Compactness for $d=4$]\label{t:compact4}
On a given smooth compact manifold $M$, 
the set of isospectral metrics whose sectional curvatures are bounded in some compact interval  
is \textbf{compact} for the $C^\infty$ topology. 
If $(M_i,g_i)$ is a family of isospectral manifolds of dimension $4$
with negative sectional curvatures bounded in some compact interval, 
then $(M_i,g_i)$ contains only \textbf{finitely many diffeomorphism types}.
\end{thm}

We would like to sketch the ideas behind the compactness results and we refer to the original papers for additional details. 
Set $n$ some positive integer and 
some real parameters 
$C,\delta,v>0 $.  
We denote by $\mathcal{M}\left(n,\delta,v,C\right)$
the set of $n$--dimensional manifolds 
with bounded sectional curvature $\vert K\vert\leqslant C$, a lower bound on the injectivity radius
$\text{inj}(M)\geqslant\delta>0$ and an upper bound on the volume
$\text{Vol}(M)\leqslant v
$.
Consider isospectral metrics $I^n(C)_{<0}$ whose
sectional curvature $\vert K\vert$ is bounded by some fixed constant $C$ and whose sectional curvature is negative.
Isospectrality and negative curvature 
ensures that Riemannian manifolds in $I^n(C)_{<0}$ have \textbf{fixed volume} since the 
volume is a spectral invariant by Weyl's law. The \textbf{length $\ell_1$ of the
shortest closed geodesic is fixed} using Duistermaat--Guillemin's trace formula. Therefore the
injectivity radius of every Riemannian manifolds in $I^n(C)_{<0}$ is uniformly bounded from below using the inequality
due to Klingenberger~\cite[p.~78]{PetersCheegerfiniteness}~:
\begin{equation}
\text{inj}(M)\geqslant \inf\left(\frac{\pi}{\sqrt{C}},\frac{1}{2}\ell_1  \right).
\end{equation}

Hence $ I^n(C)_{<0}\subset \mathcal{M}\left(n,\delta,v,C\right)$ for some real parameters 
$C,\delta,v>0 $ and by the
Cheeger finiteness Theorem in the version of Peters~\cite[Corollary 3.8 p.~9]{Boileau}~\cite[p.~77]{PetersCheegerfiniteness}, we find that~:
\begin{thm}[Cheeger finiteness for diffeomorphism types]
The Riemannian manifolds 
in $\mathcal{M}\left(n,\delta,v,C\right)$ hence 
in $ I^n(C)_{<0}$
have \textbf{finite number of diffeomorphism types}.
\end{thm}
To explain the compactness, we recall that
fixing the spectrum of $\Delta$ fixes the heat coefficients. Using some results 
of Gilkey on the structure of heat coefficients~\cite[Thm 2.1 p.~189, Lemm 3.1 p.~193]{Zhou}, Zhou proves that
for all $g\in I^n(C)_{<0}$,
its curvature $R(g)$ is bounded in all Sobolev norms $W^{k,2}(g)$ of order $k$ 
where the Sobolev norms are also defined using the same metric $g\in I^n(C)_{<0}$~\cite[Lemma 3.2 p.~193]{Zhou}.  
Then by some geometric properties of Sobolev constants, we use Sobolev inequalities
to control the $C^k$ norms of $R(g)$ in terms of the $W^{k^\prime,2}(g), k^\prime\in \mathbb{N}$ \textbf{uniformly in the metric} $g\in  I^n(C)_{<0}$.
So we converted global integrated informations on the metric and curvature into pointwise bounds on the curvature.
The conclusion now follows from the $C^k$ version of Cheeger--Gromov compactness Theorem~\cite[p.~701]{andersonisospec}~\cite[Thm 2.2 p.~190]{Zhou}~\cite[Thm $A^\prime$ p.~27]{Kasue}~:
\begin{thm}[Cheeger--Gromov $C^k$ compactness]
Fix a positive integer $k$ and $\alpha\in (0,1)$.
The space of $n$--dimensional Riemannian manifolds s.t.
$\Vert\nabla^j R\Vert_{C^0}\leqslant C$, $\forall j\leqslant k $, $\text{Vol}(M)\geqslant v>0$ and $\text{diam}_M\leqslant D$
is precompact in the $C^{k+1,\alpha}$ topology. 
More precisely
\begin{enumerate}
\item for fixed $M$, given any $\alpha<1$ and any sequence of metrics
$(g_i)_i$ on $M$
satisfying the above bounds, we can extract a convergent subsequence 
in the H\"older $C^{k+1,\alpha^\prime}$ topology for all $\alpha^\prime<\alpha$ to a limit metric $g$ of H\"older regularity $C^{k+1,\alpha}$.

\item For any sequence $\left(M_i,g_i\right)$ satisfying the above bounds, there is a subsequence which converges
\textbf{in the Lipschitz topology} to a limit metric of H\"older regularity $C^{k+1,\alpha}$ for any $0<\alpha<1$.
\end{enumerate}
\end{thm}
We use the fact that for Riemannian manifolds with metric $g\in I^n(C)_{<0}$,
the assumptions of the above Theorem are satisfied for every $k$, which explains the compactness result
of Theorem \ref{t:compact4}.

\subsubsection{Consequence of the compactness result and proof of the second claim of Theorem~\ref{t:mainthm1}.}
A sequence $(M_i,g_i)$ of Riemannian manifolds of negative curvature
s.t. $ \int_M :\phi^2(x):dv$ has fixed probability distribution 
is in fact an isospectral sequence of Riemannian manifolds. 
But since the Laplace spectrum determines the
length spectrum, the sequence $(M_i,g_i)$ of Riemannian manifolds is isospectral and 
along this sequence the geodesics of shortest length has fixed length $\ell>0$ hence 
the sequence $(M_i,g_i)$ is precompact in the sense of Anderson~\cite{andersonisospec} and by Theorem~\ref{t:compactd3}, there exists
a subsequence such that
$M_i$ has \textbf{fixed diffeomorphism type} and $g_i\rightarrow g$ to some metric $g$ in the $C^\infty$ 
topology
which is the second claim from Theorem~\ref{t:mainthm1}.
The discussion for $d=4$ is similar. A sequence $(M_i,g_i)_i$ of Riemannian manifold 
whose partition function $Z_g$ is given,
is isospectral and the condition that
the sectional curvature is in some bounded interval
$[-\varepsilon^{-1},-\varepsilon]$
implies that the sequence $(M_i,g_i)_i$ satisfies the 
assumptions of
Theorem~\ref{t:compact4}. Then the conclusion follows. 

\subsubsection{Rigidity in negative curvature.}
\label{ss:rigidity}

Up to now, we have proved that the probability distribution of 
$\int_M:\phi^2(x):dv$ determines the Laplace spectrum. Recall 
$\mathcal{R}(M)_{\leqslant -\varepsilon}$ denotes the set of isometry classes of metrics whose sectional curvatures are bounded from above by $-\varepsilon$.
To conclude the proof of claim 1) from Theorem~\ref{t:mainthm1}, it remains to show that~:
\begin{prop}\label{p:rigidity}
Let $M$ be a smooth closed compact manifold and $N$ be some finite dimensional submanifold in $\mathcal{G}\subset \mathcal{R}(M)$ such that
$\partial N\subset \partial \mathcal{G}$.
For all $\varepsilon>0$, 
the set of isospectral metrics
\begin{itemize}
\item in $N\cap \mathcal{R}(M)_{\leqslant -\varepsilon}\cap \mathcal{G}_{\geqslant\varepsilon}$ when $d\leqslant 3$,
\item in $N\cap \mathcal{R}(M)_{[-\varepsilon^{-1}, -\varepsilon]}\cap \mathcal{G}_{\geqslant\varepsilon}$ when $d=4$, 
\end{itemize}
is finite.
\end{prop}
We prove Proposition~\ref{p:rigidity} by giving 
a simple adaptation of a result due to Sarnak~\cite{sarnak1990determinants} in dimension $2$ and 
Sharafutdinov~\cite{sharafutdinov2009local} for hyperbolic metrics that 
for a finite dimensional manifold of metrics of negative curvature, 
there are only a \textbf{finite number of isospectral metrics}. 

\section{Proof of Proposition~\ref{p:rigidity}.}

In the next subsection, we introduce the geometrical tools 
needed to prove Proposition~\ref{p:rigidity}.

\subsubsection{Convergence in the space of metrics.}

Let us recall the notion of convergence in the moduli space $\mathcal{R}\left(M\right)$.
We work on a smooth closed compact manifold $M$ of dimension $d=2,3$. 
The convergence of isometry classes $[g_n]\rightarrow [g] $ means that there is a sequence of representatives
$g_n\rightarrow g$ in the $C^\infty$ topology for $2$-tensors.

\subsubsection{Symmetric tensors on Riemannian manifolds and a Hodge type decomposition of metrics.}
\label{ss:symtenshodge}
In the sequel, for any smooth vector bundle $E\mapsto M$, 
we shall use the notation $C^{k,\alpha}(E), k\in \mathbb{N}, \alpha\in(0,1)$, $H^s(E),s\in \mathbb{R}$, $C^0(E)$ to denote 
sections of $E$ of H\"older regularity $C^{k,\alpha}$, Sobolev $H^s$ and $C^0$ respectively.
Consider a Riemannian manifold $(M,g)$ and denote by $d\lambda$ the Liouville metric
on $SM$.
Consider the space of symmetric covariant $m$--tensor denoted by $S^mT^*M\subset T^mT^*M$ where $T^mT^*M$ are the
covariant $m$--tensors on $M$.
We will denote by $\sigma$ the natural symmetrization operator
acting on sections of $T^mT^*M$~\cite[p.~1267-1268]{croke1998} whose image are sections of 
$S^mT^*M$. 
The metric $g$ on $TM$ defines a canonical vertical metric on $T^mT^*M$ which induces a 
canonical $L^2$ structure on the space of smooth sections $C^\infty\left(T^mT^*M \right)$.
The Liouville measure $d\lambda$ also induces an $L^2$ structure on $C^\infty(SM)$.
Then there is a map denoted by $\pi_m^*$ going from
$C^0\left(S^mT^*M\right)$ to $C^0\left(SM\right)$ 
which identifies a symmetric $m$-tensor with a function on $SM$~: for $(x;v)\in SM$, $f\in C^0\left(S^mT^*M\right)$ we have
$  \pi_m^*f(x,v)=f(x;v,\dots,v)$ whose formal adjoint $\pi_{m*}$ with respect 
to the two $L^2$ structures defined above is defined as~: 
\begin{eqnarray}
\left\langle\pi_m^*f,u\right\rangle_{L^2(SM)}=\left\langle f,\pi_{m*}u \right\rangle_{L^2(S^*T^*M)}.
\end{eqnarray}
If $\nabla$ denotes the Levi--Civita connection, 
we define an operator
$D_g=\sigma \circ\nabla: C^\infty\left( S^mT^*M\right) \longmapsto C^\infty\left( S^{m+1}T^*M\right)$.
Its formal adjoint w.r.t. the $L^2$ scalar product on $S^mT^*M $ reads
$-D^*_g=-Tr\left( D_g\right)=-Tr\left(\nabla \right)$
where the trace is taken w.r.t. the first two factors:
$$ Tr\left(\nabla u\right)_{\alpha_1\dots \alpha_{m-1}}=\nabla^\beta u_{\beta \alpha_1\dots \alpha_{m-1}}  $$
where we sum over the repeated index $\beta$.
There is an explicit relation between
the operator $D_g$ and the generator $X\in C^\infty(SM)$ of the geodesic flow acting by Lie derivative~\cite{croke1998}~\cite[Prop 3.10 p.~28]{LefeuvreM2}~: 
\begin{equation}\label{e:identityXgeodvsD}
X\pi_m^*=\pi_{m+1}^*D_g.
\end{equation}

A fundamental result for inverse problems in negative curvature is some kind of Hodge type decomposition of 
metrics due to 
Croke--Sharafutdinov~\cite[Thm 2.2 p.~1269]{croke1998}~\cite[Thm 3.8 p.~26]{LefeuvreM2}~:
\begin{thm}\label{t:hodgemetrics}
Let $(M,g)$ be a compact Riemannian manifold s.t. the geodesic flow
on $SM$ has at least one dense geodesic and
$k\geqslant 1$ an integer. Then
every symmetric $2$--tensor $T\in H^k\left(S^2T^*M\right)$
admits the following \textbf{unique decomposition} 
\begin{equation}
T=T^s+D_g\theta,\,\ D_g^*T^s=0  
\end{equation}
where $T^s\in \ker(D_g^*)\cap H^k\left(S^2T^*M\right)$ is called the \textbf{solenoidal part} w.r.t. $g$ of the tensor~\footnote{also called divergence free part} 
and $D_g\theta=\sigma \nabla \theta$ is the \textbf{potential part} w.r.t. $g$
where $\theta \in H^{k+1}(T^*M)$ is a $1$--form,
$\nabla$ is the covariant derivative w.r.t. $g$ and $\sigma$ is the symmetrization operator. 
\end{thm}
The uniqueness of the decomposition and
$C^\infty=\bigcap_{k+1}^\infty H^k$
implies the above Theorem 
holds true for $C^\infty$ tensors. 
%
%
Geometrically, a consequence of the above Theorem is that the tangent space $T_{g}\textbf{Met}(M)\simeq C^\infty(S^2T^*M)$ to any metric $g$ whose geodesic flow admits at least one dense orbit, there is a decomposition of the form~: 
$$T_{g}\textbf{Met}(M) = \text{solenoidal tensors for }g \oplus \text{potential tensors for }g.$$

\subsection{The geometry of $\textbf{Met}(M)$ and a slice Theorem.}

In the next Theorem, we shall examine the consequences of the above Hodge type decomposition for 
the geometry of $\textbf{Met}(M)$.
In some sense, it 
shows that the space $g+\ker\left(D^*_g \right)$ of \textbf{perturbations of $g$ which are solenoidal w.r.t. $g$} is transverse, near $g$, to the
orbits of $\textbf{Diff}(M)$ and is therefore
a local slice to the orbit of $\textbf{Diff}(M)$ near $g$.
In fact, such 
result was proved by Ebin in his thesis in the ILH setting as described in 
subsubsection~\ref{ss:modulimetrics} but the slice Theorem we present 
here is more adapted to negatively curved metrics.
A consequence of the Hodge type decomposition in the space of metrics from Theorem~\ref{t:hodgemetrics} is the following~\cite[Lemma 4.1]{guillarmoulefeuvre}(see also~\cite[Thm 2.1]{croke2000local} for the boundary case)~: 
\begin{thm}\label{l:makemetsolenoidal}[Croke--Dairbekov--Sharafutdinov, Guillarmou--Lefeuvre]
Let $M$ be a compact manifold.  
For any smooth metric $g_0$ whose geodesic flow has one dense orbit, 
for every integer $k\in \mathbb{N}, k\geqslant 2 $ and real number
$\alpha\in (0,1)$, 
there exists a neighborhood $\mathcal{U}$ of $g_0$ in $C^{k,\alpha}(S^2T^*M)$
such that for any $g\in \mathcal{U}$, there is a $C^{k,\alpha}$ metric $g^\prime=\Phi^*g$ isometric to $g$
where $\Phi$ is a diffeomorphism of regularity $C^{k+1,\alpha}$ 
such that $g^\prime-g_0$ is \textbf{solenoidal} w.r.t $g_0$, moreover the map $\Psi~:g\in \mathcal{U}\subset C^{k,\alpha}(S^2T^*M) \mapsto g^\prime\in C^{k,\alpha}(S^2T^*M)$ is \textbf{smooth}.

In particular, the above Theorem holds true for $g_0$ with negative sectional curvatures.
\end{thm}

Intuitively, 
the picture one should have in mind is that in the space $\textbf{Met}(M)$ of metrics (viewed as an open cone of the space of $2$--tensors hence as a Fr\'echet manifold), the tangent space $T_{g_0}\textbf{Met}(M)$ to $g_0$ 
admits the decomposition~: 
\begin{eqnarray}\label{e:decompmetric}
T_{g_0}\textbf{Met}(M) = \underset{\text{solenoidal tensors for }g_0}{\underbrace{\ker D_{g_0}^*}} \oplus \underset{\text{potential tensors for }g_0}{\underbrace{\text{Im} D_{g_0}}}.
\end{eqnarray}
The space of potential tensors for $g_0$ is precisely the tangent part to the orbit through $g_0$ of the action of the group of diffeomorphisms which is $T_{g_0} (\textbf{Diff}(M).g_0)$. Hence starting from $g_0$ and adding a small solenoidal part exactly means moving
in the transversal direction to the orbits of $\textbf{Diff}(M)$ which means after projection that we are moving in the quotient space
$\mathcal{R}\left(M \right)=\textbf{Met}(M)/\textbf{Diff}(M)$.
Since the above 
Theorem is proved using Banach fixed point, the metric $g^\prime$ isometric to $g$ is only known to belong 
to some H\"older space $C^{k,\alpha}$ and not necessarily $C^\infty$ 
but this is sufficient for our purpose since the index $k\geqslant 2$. 
%
%
%
%

\subsection{Injectivity of the X--ray transform.}
\label{s:Xray}
Periodic orbits $\gamma$ of the vector field $X\in C^\infty(T(SM))$
which generates the geodesic flow of $g$ on $SM$
are defined as continuous maps~:
\begin{equation}
\gamma:t\in [0,T_\gamma]\longmapsto \left(\gamma(t),\dot{\gamma}(t)\right)\in SM
\end{equation} 
where $\gamma$ is parametrized at unit speed.
The closed geodesic $\gamma$ 
defines a \textbf{distribution} in $\mathcal{D}^\prime(SM)$, denoted by 
$\delta_\gamma$, as follows~:
\begin{eqnarray}
\left\langle \delta_\gamma, f \right\rangle = \int_0^{T_\gamma} f\left(\gamma(t),\dot{\gamma}(t)\right)dt. 
\end{eqnarray}
 
Recall that any symmetric $m$-tensor $f\in C^0\left(S^mT^*M\right),$ 
can be lifted as
a function on the sphere bundle $SM$ by $\pi_m^*$ defined in subsubsection~\ref{ss:symtenshodge}.
It follows that the distributions $\delta_\gamma$ act on 
symmetric $m$-tensor in $C^0\left(S^mT^*M\right)$ as follows~:
\begin{eqnarray}\label{e:deltagamma}
\left\langle \delta_\gamma, f \right\rangle = \int_0^{T_\gamma} \left(\pi_m^*f\right)\left(\gamma(t),\dot{\gamma}(t)\right)dt.
\end{eqnarray}

Now let us recall some important properties for periodic geodesics 
on manifolds with negative curvature~:
\begin{prop}\label{p:negativecurv}
Let $(M,g)$ be a compact Riemannian manifold s.t. $g$ has negative sectional
cuvatures. Denote by $\pi_1(M)$ the free homotopy classes of loops
~\footnote{The word free means that the loops are not based.} in $M$. 
Then~:
\begin{itemize}
\item the geodesic flow has the Anosov property in particular it has one everywhere dense geodesic~\cite[Thm 17.6.2]{Katokhasselblatt},
\item the 
periodic geodesics of $g$ in $SM$ are in $1-1$ correspondence
with free homotopy classes of loops in $M$~\cite[Thm 3.8.14 p.~357]{klingenberg},
\item each geodesic $\gamma$ is the unique \textbf{minimizer of the length functional}
among $C^1$ loops in the free homotopy class $[\gamma]\in \pi_1(M)$~\cite[Thm 3.8.14 p.~357]{klingenberg}.
\end{itemize}
\end{prop}

By considering the collection of 
all maps $(\delta_\gamma)_{[\gamma]\in \pi_1(M)}$, for all periodic geodesics, 
we can define the $X$-ray transform.
\begin{defi}[$X$-ray transform]
A metric $g$ with negative curvature being fixed, 
the $X$-ray transform is a linear map defined as~:
\begin{eqnarray}\label{e:Xray}
I_2: f\in C^0\left(S^2T^*M\right) \longmapsto  \left(\left\langle \delta_\gamma, f \right\rangle =\int_0^{T_\gamma} f\left(\gamma(t),\dot{\gamma}(t)\right) dt  \right)_{[\gamma]\in \pi_1(M)}
\end{eqnarray}
which maps $2$-tensors to sequences indexed by the free homotopy classes $\pi_1(M)$
of closed loops in $M$.
The map $I_2$ depends on the chosen metric $g$ since geodesics of $g$
explicitely enter in the definition of $I_2$.
\end{defi}
Note that the above $X$-ray transform is well--defined for \textbf{continuous} tensors
hence for every tensors of high enough Sobolev regularity $s>\frac{\dim(M)}{2}$
or H\"older regularity $C^{k,\alpha}$ for $k\in \mathbb{N}, \alpha\in (0,1)$.

Before we discuss the injectivity of the $X$-ray transform, we need to
introduce some geometric formalism needed in the formulation of energy identities following~\cite[section 2]{PSU15}.
The tangent bundle to $SM$ admits a direct orthogonal decomposition~:
$$T (SM ) = \mathbb{V} \perp \mathbb{H} \perp \mathbb{R}X,$$
where $\mathbb{H}$ is the horizontal bundle,  $\mathbb{V}$ is the vertical bundle, 
$X$ is the vector field generating the geodesic flow and 
$SM$ is endowed with
the Sasaki metric $\widehat{g}$ induced from the metric $g$ on the base manifold $M$. 
The Levi-Civita connection $\widehat{\nabla}$
for the Sasaki metric $\widehat{g}$ on $SM$ admits the following decomposition~:
\begin{equation}
\forall u\in C^\infty(SM), \widehat{\nabla}u=\nabla^vu+\nabla^hu+\left(Xu\right)X
\end{equation}
where $\nabla^{v,h}$ are the respective vertical and horizontal connections (the orthogonal projection of the 
connection $\widehat{\nabla}$ on the vertical and horizontal bundles), i.e. 
$\nabla^vu\in \mathbb{V}, \nabla^hu\in \mathbb{H} $.

An important result 
about $I_2$ reads~:
\begin{thm}[Injectivity of the X ray transform]
\label{t:injectXray}
Let $k\in \mathbb{N}, k\geqslant 2, \alpha\in (0,1)$.
Let $g$ be a smooth metric with negative curvature.
Then the $X$-ray transform $I_2$ defined above
restricted to solenoidal tensors in $\ker(D_g^*)$ 
of H\"older regularity
$C^{k,\alpha}$ is \textbf{injective}.
\end{thm}
The above result was proved in the $C^\infty$ case by Croke--Sharafutdinov and is well--known in 
H\"older regularity $C^{k,\alpha}$ although we could not find a reference.
In the present work, we shall need the injectivity of $I_2$ acting on $g$--solenoidal
tensors of regularity $C^{k,\alpha}$ due to the loss of regularity caused by
applying the slice Theorem~\ref{l:makemetsolenoidal}.
We slightly adapt the proof following notes of Lefeuvre which are themselves heavily based on the original 
proof of~\cite{croke1998}.
\begin{proof}
The proof in~\cite{croke1998} relies on the following ingredients~:
\begin{enumerate} 
\item The Hodge like decomposition~\ref{e:decompmetric} which is well--defined in every Sobolev regularity $H^s, s\in \mathbb{R}$ or H\"older regularity $C^{k,\alpha}$~\footnote{since it relies on ellipticity of $D_g$ and pseudodifferential calculus}.
\item 
The Pestov identity
\begin{equation}
\Vert \nabla^vXu\Vert^2\geqslant \Vert X\nabla^vu\Vert^2 +d\Vert Xu\Vert^2
\end{equation}
which is well--defined for Sobolev tensors in $H^s(S^2T^*M)$
for $s\geqslant 2$ and another energy identity if $X^2u\in H^0\left(S^{m+1}T^*M \right)$ then
\begin{equation}
\Vert X\nabla^vu\Vert^2-\Vert \nabla^vXu\Vert^2=\Vert \nabla^hu\Vert^2+ ((m+d)(m+1)-d)\Vert Xu\Vert^2+ \Vert \nabla^vXu\Vert^2
\end{equation}
which is also valid for Sobolev tensors in $H^s(S^2T^*M)$
for $s\geqslant 2$. In particular, both identities are valid for tensors $u$ of regularity at least $C^2\subset H^2$ 
hence for tensors of H\"older regularity $C^{k,\alpha}$ $k\in \mathbb{N}, k\geqslant 2, \alpha\in (0,1)$.
\item The Livsic Theorem in regularity $C^k$ for all $k\in \mathbb{N}$ due to De La Llave--Marco--Moriyon~\cite[Remark p.~578]{DeLlaveMarcoMoriyon} which states that for any function $\tilde{f}\in C^k(SM)$ s.t. $ \forall [\gamma]\in \pi_1(M),\,\ \left\langle \delta_\gamma, f \right\rangle =0$, there exists $u\in C^k(SM)$ s.t. $\tilde{f}=Xu$. 
\end{enumerate}

Fix $k\geqslant 2$ so that we can use the energy identities.
The proof is almost verbatim the one of Croke--Sharafutdinov~\cite[Thm 1.3]{croke1998} using a $C^k$ version of 
Livsic Theorem instead of the $C^\infty$ version.
Assume $f\in C^{k,\alpha}(S^2T^*M)$ is solenoidal and that $I_2f=0$. Then by the inclusion 
$C^{k,\alpha}\subset C^k$ and the $C^k$ version of Livsic Theorem,
there exists $u\in C^k(SM)$ s.t. $ \pi_2^*f=Xu $. Thus $X^2u=X\pi_2^*f=\pi_3^*D_gf$ by equation~\ref{e:identityXgeodvsD} since 
$D_gf$ is trace free because $D_g^*f=0=Tr(D_gf)$.
Combining the above energy identities yields~:
\begin{eqnarray*}
0\geqslant \underset{\text{by Pestov}}{\underbrace{-d\Vert Xu\Vert^2\geqslant \Vert\nabla_X\nabla^vu\Vert^2-\Vert\nabla^v\nabla_Xu\Vert^2}}
\underset{\text{by the second energy identity}}{\underbrace{= \Vert \nabla^hu\Vert^2+ ((2+d)(2+1)-d)\Vert Xu\Vert^2+ \Vert \nabla^vXu\Vert^2\geqslant 0}} 
\end{eqnarray*}
which implies that $Xu=0$.
\end{proof}

In the sequel, for 
every negatively curved metric $g$, 
we denote by $\ell_g(\gamma)$ the length of the unique closed geodesic
$\gamma$ given a class $[\gamma]\in \pi_1(M)$. 

There is a natural map from moduli space of metrics of negative curvature to 
periods 
\begin{eqnarray}\label{e:periodmap}
[g]\in \tilde{N}\in \mathcal{R}(M)_{<0}\longmapsto \mathcal{P}([g])=\{ \ell_g(\gamma) ; [\gamma]\in \pi_1(M) \} \subset \mathbb{R}_{>0} .
\end{eqnarray}
 
The assignment $g\mapsto \ell_g(\gamma)$ depends nonlinearly 
on the metric $g$. However a classical
observation which can be found in~\cite{GuiKaz} is the relation between the differential
of the length function and the X--ray transform, for any metric $g\in \textbf{Met}(M)$ and symmetric $2$--tensor $h$, the differential
of $\ell$ at $g$ in the direction $h$ reads~:
\begin{eqnarray}
D\ell_{g}(\gamma)(h)=\frac{1}{2}
I_2(h)_{[\gamma]}.
\end{eqnarray}
Therefore, one should think about the X--ray transform $I_2$ as
a linearized version of the length function and the injectivity of $I_2$
reflects the injectivity of the nonlinear map~\ref{e:periodmap}.

After these rather long geometric preparations, we can proceed to prove 
Proposition~\ref{p:rigidity}. 
\subsection{Proof of Proposition~\ref{p:rigidity} by a contradiction argument.}  

In dimension $d\leqslant 3$ (resp $d=4$),
we assume by contradiction that the set of isospectral metrics in
$N\cap \mathcal{R}\left(M\right)_{\leqslant-\varepsilon}\cap \mathcal{G}_{\geqslant\varepsilon}$ 
(resp $N\cap \mathcal{R}\left(M\right)_{[-\varepsilon^{-1},-\varepsilon]}\cap \mathcal{G}_{\geqslant\varepsilon}$) 
has an infinite number of classes.
Therefore, we assume there exists an infinite sequence $(g^\prime_n)_n$ of smooth isospectral 
metrics on $M$
whose isometry classes $([g^\prime_n])_n$ are $2$ by $2$ distinct in $N\cap \mathcal{R}\left(M\right)_{\leqslant-\varepsilon}\cap \mathcal{G}_{\geqslant\varepsilon}$ (resp 
$N\cap \mathcal{R}\left(M\right)_{[-\varepsilon^{-1},-\varepsilon]}\cap \mathcal{G}_{\geqslant\varepsilon}$). 
So if $(g^\prime_n)_n$ is a sequence of isospectral metrics of negative curvature $\leqslant-\varepsilon$ (resp in $[-\varepsilon^{-1},-\varepsilon]$), the above compactness
Theorems~\ref{t:compactd2},~\ref{t:compactd3},~\ref{t:compact4} tell us that we may extract a subsequence such that $g^\prime_n\rightarrow g$ in the $C^\infty$--topology for some metric
$g$ of negative curvature $\leqslant -\varepsilon$ (resp in $[-\varepsilon^{-1},-\varepsilon]$). It is important to note that
the limit metric $g$ has no isometry group since its class $[g]$ belongs to $\mathcal{G}_{\geqslant \varepsilon}$ and $[g]$ belongs to $N$
since $N\cap \mathcal{R}\left(M\right)_{\leqslant-\varepsilon}\cap \mathcal{G}_{\geqslant\varepsilon} $ (resp $N\cap \mathcal{R}\left(M\right)_{[-\varepsilon^{-1},-\varepsilon]}\cap \mathcal{G}_{\geqslant\varepsilon}$) is \textbf{closed}
from the condition $\partial N\subset \partial\mathcal{G}$.
In dimension $3$, we can apply the compactness Theorem~\ref{t:compactd3} since the spectrum determines
the length of the shortest closed geodesic by Theorem~\ref{t:DGtrace}.

%

Now since $g$ has negative curvature, we can make use of the slice Theorem~\ref{l:makemetsolenoidal} and produce a new sequence $(g_n)_n$ of metrics with the following properties~:
\\
\\
\fbox{
\begin{minipage}{0.94\textwidth}
\begin{coro}\label{c:coroisometric}
There exists a sequence of metrics $(g_n=\Psi\left( g^\prime_n\right))_n$ of regularity $C^{k,\alpha}, k\geqslant 2, \alpha\in (0,1)$ such 
that $[g_n]=[g^\prime_n], \forall n\in \mathbb{N}$,
the difference $\varepsilon_n=g_n-g\in C^{k,\alpha}(S^2T^*M) \rightarrow 0$ 
is \textbf{solenoidal} w.r.t. $g$, the metrics $g_n$ all have the \textbf{same length spectrum}.
\end{coro}
\end{minipage}
}
\\
\\
It is important to note that it is no longer 
a priori true that the sequence $(g_n)_n$ is made of smooth metrics.
The \textbf{solenoidal property} 
will be very important in the sequel since we shall
use the injectivity of the X-ray transform for solenoidal tensors w.r.t. $g$.
We assume by contradiction that the sequence of metrics $g_n$ is non stationary, 
which means that the sequence $\varepsilon_n=g_n-g_0$ 
never vanishes for every $n$ and
$\varepsilon_n\rightarrow 0$
in $C^{k,\alpha}(S^2T^*M)$.
By Proposition~\ref{p:negativecurv}, for each metric $g_n$ and for every class $[\gamma]\in \pi_1(M)$, 
we shall denote by $\gamma_n$ (resp $\gamma$) the unique geodesic representative of $[\gamma]$ 
in $\pi_1\left(M\right)$
for the metric $g_n$ (resp for the metric $g$).

For each closed curve $\gamma$ in $SM$, 
we define a Radon measure $\delta_\gamma\in \mathcal{D}^\prime(SM)$ by equation~(\ref{e:deltagamma}) in subsection~\ref{s:Xray}.
Recall that a sequence $\mu_n$ of Radon measures on $SM$ is said to weakly--$*$ converge to $\mu$ if
for every continuous $\varphi\in C^0(SM)$, $\mu_n(\varphi)\rightarrow \mu(\varphi)$ when $n\rightarrow +\infty$. 
By Proposition~\ref{p:convradon} proved in
the appendix, we have the weak--$*$ convergence $\delta_{\gamma_n}\rightarrow \delta_\gamma$ of the 
Radon measures on $SM$.
By the convergence of metrics $g_n\rightarrow g$,
for every free homotopy class $[\gamma]$ in $\pi_1\left(M\right)$, for every $n$, we have the
convergence $\ell_{g_n}(\gamma_n)\rightarrow \ell_g(\gamma)$ by~\cite[Lemma 4.1 p.~11]{dang2018fried}.

\subsubsection{Inequalities satisfied by $\varepsilon_n$.}

From the fact that 
the metrics are isospectral and the length spectrum is discrete, we deduce that
$\ell_{g_n}(\gamma_n)=\ell_g(\gamma)$ for every $n\geqslant
N_\gamma$ where the integer $N_\gamma$ depends on $[\gamma]\in \pi_1(M)$.
By equation~(\ref{e:deltagamma}) defining the Radon measures
$\delta_\gamma\in \mathcal{D}^\prime(SM)$ 
carried by closed curves $\gamma$, 
the length of the curve $\gamma$ for the metric $g$ is defined as
$\ell_g(\gamma)=\delta_\gamma\left( g\right)$.
The metrics $g_n$ have negative sectional curvatures hence by Proposition~\ref{p:negativecurv}, 
every closed geodesic 
$\gamma_n$ is minimizing for $g_n$ in the class $[\gamma]\in \pi_1(M)$
which implies the 
inequality 
$\delta_{\gamma_n}\left( g_n \right) \leqslant \delta_\gamma\left( g_n \right) $.
But for $n\geqslant N_\gamma$, we find
that $\delta_{\gamma_n}\left( g_n \right)=\delta_\gamma\left( g_0 \right) $
from which we deduce the inequality
$\delta_\gamma\left( g_0 \right) \leqslant \delta_\gamma\left( g_n \right) $ 
which implies
\begin{eqnarray}\label{e:firstineq}
\delta_\gamma\left(\varepsilon_n\right) \geqslant 0, \forall n\geqslant N_\gamma.
\end{eqnarray} 
Conversely since $\gamma$ minimizes the length for $g_0$ we have a reverse inequality  
$\delta_{\gamma_n}\left( g_n \right)= \delta_{\gamma}\left( g_0 \right)\leqslant \delta_{\gamma_n}\left( g_0 \right)$
which implies the second inequality~:
\begin{eqnarray}\label{e:secineq}
\delta_{\gamma_n}\left( \varepsilon_n\right)\leqslant 0, \forall n\geqslant N_\gamma.
\end{eqnarray}

\subsubsection{Another compactness argument and conclusion of the proof of Proposition~\ref{p:rigidity}.}

Now we would like to know if we can extract a subsequence 
from $\frac{\varepsilon_n}{\Vert \varepsilon_n\Vert_\infty}\in C^{k,\alpha}(S^2T^*M)$
\textbf{with non trivial limit} so that we obtain inequalities on the $X$-ray transform 
which are independent of $n$. 
We denote by $\textbf{Met}^{k,\alpha}(M)$
the set of Riemannian metrics of H\"older regularity $C^{k,\alpha}$.
The assumption in corollary~\ref{c:coroisometric} 
means that 
there exists an abstract smooth
submanifold $\tilde{N}$ of the same dimension as $N$ near $g$ 
which contains the sequence $(g_n)_n$
and a $C^\infty$- map
$ \iota~: \tilde{N}\longmapsto \textbf{Met}^{k,\alpha}(M) $
such that $\pi\circ\iota(\tilde{N})=N$ near $[g]$ where $\pi:\textbf{Met}^{k,\alpha}(M)\mapsto \mathcal{R}(M)$
is the natural projection induced by the quotient map. 
One should think of 
$\tilde{N}$ as some kind of cover of $N$ near $[g]$.

Choose any smooth Riemannian metric $\tilde{g}$ on the finite dimensional
submanifold $\tilde{N}$ and denote by $v_n$
a sequence of tangent vectors in $T_{g_0}\tilde{N}$ such 
that
$\iota\left(\exp_{g_0}(v_n)\right)=g_0+\frac{\varepsilon_n}{\Vert \varepsilon_n \Vert_\infty}$ where
$\exp$ is the Riemannian exponential map induced by the metric $\tilde{g}$.
Since the exponential map $v\in T_{g_0}\tilde{N} \mapsto \exp_{g_0}\left(v \right)$ is 
a diffeomorphism 
near the origin whose differential at $0$ is the identity, we may find 
that the norm of the sequence $v_n$ is equivalent to
the distance $\text{dist}\left(g_0+\frac{\varepsilon_n}{\Vert \varepsilon_n \Vert_\infty},g_0\right)$.
Since $\tilde{N}$ has finite dimension and the sequence $\frac{\varepsilon_n}{\Vert \varepsilon_n \Vert_\infty}$
has sup norm $1$, the sequence of tangent vectors $v_n$ is contained in some closed bounded subset
of $T_{g_0}\tilde{N}$ which avoids $0$. 
Then by compactness of closed bounded subsets 
in \textbf{finite dimension}, we can extract a 
subsequence of
$(v_n)_n$ s.t. $v_n\rightarrow v_\infty\neq 0\in T_{g_0}\tilde{N}$.
So along this subsequence, $\frac{\varepsilon_n}{\Vert\varepsilon_n \Vert}$ has a nontrivial limit
$u=\exp_{g_0}(v_\infty)\in C^{k,\alpha}(S^2T^*M)$.
 Hence, up to 
extracting a subsequence, we may assume that
$ \frac{\varepsilon_n}{\Vert\varepsilon_n\Vert_\infty} \underset{n\rightarrow \infty}{\longrightarrow} u\in C^{k,\alpha}(S^2T^*M) $ 
in the $C^{k,\alpha}$ 
topology where $u \neq 0$ and $\Vert u\Vert_\infty=1$.

Passing to the limit in both 
inequalities \ref{e:firstineq} and \ref{e:secineq} and using the fact
that the Radon measures  $\delta_{\gamma_n}$ weakly--$*$ converges to $\delta_\gamma$,  
we find that the limit $u$ satisfies
$ I_2(u)_\gamma= \delta_\gamma\left( u \right) \geqslant 0  $ and $I_2(u)_\gamma =\delta_\gamma\left( u \right)\leqslant 0 $
hence for any free homotopy class $[\gamma]\in \pi_1(M)$, 
$ I_2(u)_\gamma =\delta_\gamma\left( u \right)=0$. The above means that $u$ is a $2$--tensor which 
belongs to the kernel of the linear map $I_2$. 
But since $u$ is solenoidal w.r.t. $g$ and the restriction of 
$I_2$ to solenoidal tensors w.r.t. $g$ is \textbf{injective} by Theorem~\ref{t:injectXray},  we conclude that $u=0$
which contradicts $u\neq 0 $ in
$C^{k,\alpha}(S^2T^*M)$. 

%

\section{Appendix.}

\subsection{Gohberg--Krein's determinants.}
\label{s:Hadamard}
Set $p=[\frac{d}{2}]+1$ and let $A$ belong to the Schatten ideal
$\mathcal{I}_p$. 
Following~\cite[chapter 9]{Simon-traceideals}, 
we shall summarize the main properties of 
Gohberg--Krein's determinants and their relation with functional traces~:
\begin{lemm}[Gohberg--Krein's determinants and functional traces]
\label{l:detregtraces}
For all $A\in \mathcal{I}_p$, the Gohberg--Krein determinant $\det_p(1+zA)$ is an \textbf{entire
function} in $z\in \mathbb{C}$ and is related to traces $Tr(A^n)$ for $n>\frac{d}{2}$ by the following
formulas~:
\begin{eqnarray*}
\boxed{\text{det}_p(1+zA)= \exp\left(\sum_{n=p}^\infty  \frac{(-1)^{n+1}z^n}{n}Tr(A^n)  \right)=\prod_k\left[(1+z\lambda_k(A))\exp\left( \sum_{n=1}^{p-1}(-1)^nn^{-1}\lambda_k(A)^n  \right)\right]}
\end{eqnarray*}
where the series $\exp\left(\sum_{n=p}^\infty  \frac{(-1)^{n+1}z^n}{n}Tr(A^n)  \right)$
converges when $\vert z\vert<\Vert A\Vert_{\mathcal{I}_p}$ and the infinite product vanishes exactly when 
$z\lambda_k(A)=-1$ with multiplicity.
\end{lemm}

\subsection{Convergence of Radon measures corresponding to closed geodesics.}

The goal of this paragraph is to show that if
$g_n\mapsto g$ in the metrics of negative curvature, then
for every free homotopy class $[\gamma]\in \pi_1(M)$, denote by $\gamma_n$ (resp $\gamma$) the unique
corresponding sequence of closed geodesic for $g_n$ (resp $g$), 
the sequence of Radon measures $\delta_{\gamma_n}$ weak--$*$ converges to $\delta_\gamma$. 
We shall use the structural stability result
of Anosov flows in the version of De La Llave--Marco--Moriyon~\cite[Thm A.2 p.~598]{DeLlaveMarcoMoriyon}. 
\begin{thm}[Structural stability]\label{t:ssLlave}
Let $(M,g)$ be a Riemannian manifold of negative curvature and set $\mathcal{M}=SM$ 
to be the sphere bundle of $M$.
We denote by $X\in C^1(T\mathcal{M})$ the geodesic vector field of the metric $g$ and by 
$C^0_X(\mathcal{M},\mathcal{M})$ the space of homeomorphisms from $\mathcal{M}$ to $\mathcal{M}$ which are $C^1$ along integral curves of $X$ and $C^0(\mathcal{M})$
denotes continuous functions on $\mathcal{M}$.
Then there exists a $C^1$ neighborhood $\mathcal{U}$ 
of $X$, a submanifold $\mathcal{N}\subset C^0_X(\mathcal{M},\mathcal{M})$ and a $C^1$ map~: 
\begin{eqnarray}
S:\mathcal{U} \longmapsto \mathcal{N} \times C^0(\mathcal{M}) \\
 Y\longmapsto \left(\Phi_Y,h_Y\right)
\end{eqnarray}
satisfying the structure equation~:
\begin{equation}
\label{e:DMMstruceq}
\boxed{(\Phi_Y^{-1*}h_Y)Y=\Phi_{Y*}X}
\end{equation}
where $\left(\Phi_X,h_X\right)=(Id,1)\in C^0_X(\mathcal{M},\mathcal{M})\times C^0(M)$.
\end{thm}

The equation \ref{e:DMMstruceq} follows from
~\cite[equation (e) p.~592]{DeLlaveMarcoMoriyon}
\begin{equation}
D \Phi_Y\left(x,v\right)\left(X(x,v)\right)=h_Y\left(x,v \right)Y\left(\Phi_Y(x,v)\right),
\end{equation}
this implies that
$ D \Phi_Y\left(\Phi_Y^{-1}(x,v)\right)\left(X(\Phi_Y^{-1}(x,v))\right)=h_Y\left(\Phi_Y^{-1}(x,v) \right)Y\left(x,v\right) $ hence 
$\Phi_{Y*}X=\left(\Phi_Y^{-1*}h_Y \right) Y $.
The above equation means 
that flows in a neighborhood $\mathcal{U}$ of $X$ are conjugated to the flow generated by $X$ up to reparametrization of time, more precisely
let $\varphi^t_Y:\mathcal{M}\mapsto \mathcal{M}$ denotes the flow generated by $Y\in \mathcal{U}\subset C^1(T\mathcal{M})$, then 
there exists $\tau_Y\in C^0(\mathbb{R}\times \mathcal{M})$ s.t.~:
\begin{eqnarray}
\varphi^{t}_Y(x,v)=\Phi_Y\circ \varphi_X^{\tau_Y(t,x,v)} \circ \Phi_Y^{-1}(x,v)
\end{eqnarray}
where $\tau_Y(t,x,v)\rightarrow t$ in $C^0([0,T]\times \mathcal{M})$ for all $T>0$ 
when $Y\rightarrow X$ in $C^1(T\mathcal{M})$.

A corollary of the above result 
\begin{prop}[Convergence result for Radon measures.]
\label{p:convradon}
Under the assumptions 
of Theorem~\ref{t:ssLlave}.
Let $(g_n)_n$ be a sequence of metrics of negative curvature
which converges to $g$ 
in the $C^2$ topology.
We denote by $X\in C^1(T\mathcal{M})$ (resp $X_n\in C^1(T\mathcal{M})$) the geodesic vector field of the metric $g$
(resp $g_n$).
 
Then $X_n$ is a sequence of vector fields which converges to $X$ in $C^1(\mathcal{M})$
where
for every free homotopy class $[\gamma]\in \pi_1\left(M \right)$,
there exists $N_\gamma\in \mathbb{N}$ and a unique subsequence of periodic orbits $(\gamma_n)_{n\geqslant N_\gamma} $ of the vector field $X_n$
which converges to a periodic orbit $\gamma$ of $X$.
The corresponding Radon measures $\delta_{\gamma_n}, n\geqslant N_\gamma$ 
will weak--$*$ converge to the limit  Radon measure
$\delta_\gamma$.

In particular for every $2$--tensor $h\in C^0(S^2T^*M)$, 
recall $\pi_2^*:C^0(S^2T^*M)\mapsto C^0(SM)$, then
$$ \delta_{\gamma_n}(\pi_2^*h)\rightarrow \delta_\gamma(\pi_2^*h) . $$
\end{prop}
\begin{proof}
Let $f\in C^0(SM)$ be a continuous test function. Denote by $\varphi_n^t$ (resp $\varphi^t$) the flow
generated by $X_n$ (resp $X$) on $SM$.
By definition
$\delta_{\gamma_n}\left(f\right)=\int_0^{\ell_{g_n}(\gamma_n)} f\circ\varphi_n^t(x_n,v_n)dt$
for any $(x_n,v_n)\in\gamma_n$.
The existence of the sequence 
$\gamma_n\rightarrow \gamma$ is a simple
consequence of structural stability.
Let $\Phi_n\in C^0_X(M,M)$ denotes the sequence of
homeomorphisms
conjugating the two flows whose existence comes from Theorem~\ref{t:ssLlave}~: 
$$\varphi^t_n(x,v)= \Phi_n \circ \varphi^{\tau_n(t,x,v)}\circ \Phi_n^{-1} (x,v) $$
where $\tau_n(t,x,v) \rightarrow t$ uniformly on $[0,T]\times SM$ for all $T>0$ and 
$\Phi_n\rightarrow Id$ in $C^0(SM)$. Therefore for every $(x,v)$ on the periodic orbit
$\gamma$, the sequence $\left(x_n,v_n\right)=\Phi_n\left(x,v\right)$
lies in the periodic orbit $\gamma_n$ by structural stability and converges to $(x,v)$.
It follows that when $n\rightarrow +\infty$,
$$\delta_{\gamma_n}\left(f\right)=\int_0^{\ell_{g_n}(\gamma_n)} f\circ\varphi_n^t(x_n,v_n)dt=\int_0^{\ell_{g_n}(\gamma_n)} f\circ\Phi_n \circ \varphi^{\tau_n(t,x_n,v_n)}(x,v) dt\underset{n\rightarrow +\infty}{\rightarrow} \int_0^{\ell_g(\gamma)} f \circ \varphi^{t}(x,v) dt$$
by dominated convergence and
since the periods
$\ell_{g_n}(\gamma_n) \underset{n\rightarrow +\infty}{\rightarrow} \ell_g(\gamma)$
converge~\cite[Lemma 4.1 p.~11]{dang2018fried} and $\frac{1}{2}\ell_g(\gamma) \leqslant \ell_{g_n}(\gamma_n)\leqslant 2\ell_g(\gamma) $
for all $n\geqslant N_\gamma$~\cite[Remark 3]{dang2018fried}. 
It follows that the sequence
of Radon measures $\delta_{\gamma_n}$ will weak--$*$ converge to the limit Radon measures
$\delta_\gamma$.
\end{proof}

\appendix


\begin{thebibliography}{10}

\bibitem{arnaudonthalm}
Arnaudon, Marc, and Anton Thalmaier. 
\newblock Brownian motion and negative curvature.
\newblock Random walks, boundaries and spectra. Springer, Basel, 2011. 143-161.

\bibitem{BergerEbin}
Berger, Marcel, and D. Ebin. 
\newblock Some decompositions of the space of symmetric tensors on a Riemannian manifold.
\newblock J. Diff. Geom 3.3-4 (1969): 379-392.

\bibitem{Bourguignon}
Bourguignon, Jean-Pierre. 
\newblock Une stratification de l'espace des structures riemanniennes.
\newblock Compositio Mathematica 30.1 (1975): 1-41.

\bibitem{CDV1}
Colin de Verdi\`ere, Yves. 
\newblock Spectre du Laplacien et longueurs des g\'eod\'esiques p\'eriodiques. I. 
\newblock Compositio Mathematica 27.1 (1973): 83-106.

\bibitem{CDV2}
Colin de Verdi\`ere, Yves. 
\newblock Spectre du Laplacien et longueurs des g\'eod\'esiques p\'eriodiques. II. 
\newblock Compositio Mathematica 27.2 (1973): 159-184.

\bibitem{Ebin1}
Ebin, D. 
\newblock The manifold of Riemannian metrics. 
\newblock 1970 Global Analysis (Proc. Sympos. Pure Math., Vol. XV, Berkeley, Calif., 1968) pp. 11–40 Amer. Math. Soc., Providence, RI (1968).

\bibitem{Ebinthesis}
Ebin, David Gregory.
\newblock On the space of Riemannian metrics.
\newblock Diss. Massachusetts Institute of Technology, 1967.

\bibitem{MarsdenEbinFischer}
Marsden, Jerrold E., David G. Ebin, and Arthur E. Fischer. 
\newblock Diffeomorphism groups, hydrodynamics and relativity. 
\newblock Proceedings of the 13th Biennial Seminar of Canadian Mathematical Congress. Canadian Mathematical Congress, p.~135--279. ISBN 9780919558038 (1972)

\bibitem{Fischer}
Fischer, Arthur Elliot. 
\newblock The theory of superspace.
\newblock Relativity. Springer, Boston, MA, 1970. 303-357.

\bibitem{Gilkey}
P~Gilkey.
\newblock Invariance Theory, the Heat Equation and the Atiyah--Singer Index
Theorem, (1995).
\newblock {\em Studies in Advanced Mathematics, CRC Press, Inc}.

\bibitem{GuiKaz}
Guillemin, Victor, and David Kazhdan.
\newblock Some inverse spectral results for negatively curved 2-manifolds.
\newblock Topology 19.3 (1980): 301-312.

\bibitem{kang2017calculus}
Nam-Gyu Kang and Nikolai~G Makarov.
\newblock Calculus of conformal fields on a compact Riemann surface.
\newblock {\em arXiv preprint arXiv:1708.07361}, 2017.

\bibitem{Katokhasselblatt}
Katok, Anatole, and Boris Hasselblatt.
\newblock Introduction to the modern theory of dynamical systems.
\newblock Vol. 54. Cambridge university press, 1995.

\bibitem{klingenberg}
Klingenberg, Wilhelm PA.
\newblock Riemannian geometry (second edition).
\newblock Vol. 1. Walter de Gruyter, 1995.


\bibitem{guillarmoulefeuvre}
Guillarmou, Colin, and Thibault Lefeuvre. 
\newblock The marked length spectrum of Anosov manifolds.
\newblock Annals of Math 190 (2019), no 1. 

\bibitem{guillarmoupolyakov}
Colin Guillarmou, R{\'e}mi Rhodes, and Vincent Vargas.
\newblock Polyakov's formulation of $2d$ bosonic string theory.
\newblock {\em arXiv preprint arXiv:1607.08467}, 2016.

\bibitem{DangZhang}
Dang, Nguyen Viet, and Bin Zhang. 
\newblock Renormalization of Feynman amplitudes on manifolds by spectral zeta regularization and blow-ups. 
\newblock arXiv preprint arXiv:1712.03490 (2017), to appear in Journ. Eur. Math. Soc.

\bibitem{dimock2004markov}
J~Dimock.
\newblock Markov quantum fields on a manifold.
\newblock {\em Reviews in Mathematical Physics}, 16(02):243--255, 2004.

\bibitem{dubedatsle}
Julien Dub{\'e}dat.
\newblock SLE and the free field~: partition functions and couplings.
\newblock {\em Journal of the American Mathematical Society}, 22(4):995--1054,
  2009.

\bibitem{axelrodsinger2}
Scott Axelrod, IM~Singer, et~al.
\newblock Chern-simons perturbation theory. II.
\newblock {\em Journal of Differential Geometry}, 39(1):173--213, 1994.

\bibitem{kontsevich1994feynman}
Maxim Kontsevich.
\newblock Feynman diagrams and low-dimensional topology.
\newblock In {\em First European Congress of Mathematics Paris, July 6--10,
  1992}, pages 97--121. Springer, 1994.

\bibitem{dijkgraaf1994perturbative}
Robbert Dijkgraaf.
\newblock Perturbative topological field theory.
\newblock In {\em String Theory, Gauge Theory and Quantum Gravity'93}, pages
  189--227. World Scientific, 1994.

\bibitem{baez2000introduction}
John~C Baez.
\newblock An introduction to spin foam models of BF theory and quantum gravity.
\newblock In {\em Geometry and quantum physics}, pages 25--93. Springer, 2000.

\bibitem{witten19882+}
Edward Witten.
\newblock 2+ 1 dimensional gravity as an exactly soluble system.
\newblock {\em Nuclear Physics B}, 311(1):46--78, 1988.

\bibitem{carlip2003quantum}
Steven Carlip and Steven~Jonathan Carlip.
\newblock {\em Quantum gravity in 2+ 1 dimensions}, volume~50.
\newblock Cambridge University Press, 2003.

\bibitem{mnev2008discrete}
Cattaneo, Alberto S., Pavel Mnev, and Nicolai Reshetikhin. 
\newblock A cellular topological field theory.
\newblock arXiv preprint arXiv:1701.05874 (2017).

\bibitem{belkale2003periods}
Prakash Belkale and Patrick Brosnan.
\newblock Periods and Igusa local zeta functions.
\newblock {\em International Mathematics Research Notices},
  2003(49):2655--2670, 2003.

\bibitem{Bogner2009periods}
Christian Bogner and Stefan Weinzierl.
\newblock Periods and Feynman integrals.
\newblock {\em Journal of Mathematical Physics}, 50(4):042302, 2009.

\bibitem{DappEucl}
Dappiaggi, Claudio, Nicol\`o Drago, and Paolo Rinaldi. 
\newblock The algebra of Wick polynomials of a scalar field on a Riemannian manifold.
\newblock arXiv preprint arXiv:1903.01258 (2019).

\bibitem{DangHersc}
Dang, Nguyen Viet, and Estanislao Herscovich. 
\newblock Renormalization of quantum field theory on Riemannian manifolds.
\newblock Reviews in Mathematical Physics (2018): 1950017.

\bibitem{segal2004definition}
Graeme Segal.
\newblock The definition of conformal field theory. Topology, Geometry and
Quantum Field Theory, 421--577.
\newblock {\em London Math. Soc. Lecture Note Ser}, 308, 2004.

\bibitem{Seilergauge}
Seiler, Erhard.
\newblock Gauge theories as a problem of constructive quantum field theory and statistical mechanics.
\newblock Springer, 1982.

\bibitem{stolz2004elliptic}
Stephan Stolz and Peter Teichner.
\newblock What is an elliptic object?
\newblock {\em London Mathematical Society Lecture Note Series}, 308:247, 2004.

\bibitem{kandel2015functorial}
Santosh Kandel.
\newblock Functorial Quantum Field Theory in the Riemannian setting.
\newblock {\em arXiv preprint arXiv:1502.07219}, 2015.

\bibitem{stolz2014lecture}
Stephan Stolz.
\newblock Lecture notes: Functorial Field Theories and Factorization Algebras,
  2014.

\bibitem{sarnak1990determinants}
Peter Sarnak.
\newblock Determinants of Laplacians; heights and finiteness.
\newblock In {\em Analysis, et cetera}, pages 601--622. Elsevier, 1990.

\bibitem{BGM}
Marcel Berger, Paul Gauduchon, and Edmond Mazet.
\newblock {\em Le spectre d'une vari{\'e}t{\'e} riemannienne}.
\newblock Springer, 1971.

\bibitem{berger2012panoramic}
Marcel Berger.
\newblock {\em A panoramic view of Riemannian geometry}.
\newblock Springer Science \& Business Media, 2012.

\bibitem{moretti1999}
Valter Moretti.
\newblock One-loop stress-tensor renormalization in curved background: the
  relation between $\zeta$-function and point-splitting approaches, and an
  improved point-splitting procedure.
\newblock {\em Journal of Mathematical Physics}, 40(8):3843--3875, 1999.

\bibitem{moretti1999local}
Valter Moretti.
\newblock Local $\zeta$-function techniques vs. point-splitting procedure: a
  few rigorous results.
\newblock {\em Communications in mathematical physics}, 201(2):327--363, 1999.

\bibitem{moretti2003}
Valter Moretti.
\newblock Comments on the stress-energy tensor operator in curved spacetime.
\newblock {\em Communications in Mathematical Physics}, 232(2):189--221, 2003.

\bibitem{moretti2011}
Valter Moretti.
\newblock Local $\zeta$-functions, stress-energy tensor, field fluctuations,
  and all that, in curved static spacetime.
\newblock In {\em Cosmology, Quantum Vacuum and Zeta Functions}, pages
  323--332. Springer, 2011.

\bibitem{hack2012stress}
Thomas-Paul Hack and Valter Moretti.
\newblock On the stress--energy tensor of quantum fields in curved
  spacetimes—comparison of different regularization schemes and symmetry of
  the Hadamard/Seeley--DeWitt coefficients.
\newblock {\em Journal of Physics A: Mathematical and Theoretical},
  45(37):374019, 2012.

\bibitem{sanders2017local}
Ko~Sanders.
\newblock Local versus global temperature under a positive curvature condition.
\newblock In {\em Annales Henri Poincar{\'e}}, volume~18, pages 3737--3756.
  Springer, 2017.

\bibitem{leJan}
Yves Le~Jan.
\newblock {\em Markov Paths, Loops and Fields: {\'E}cole D'{\'E}t{\'e} de
  Probabilit{\'e}s de Saint-Flour XXXVIII--2008}, volume 2026.
\newblock Springer Science \& Business Media, 2011.

\bibitem{lawler2018topics}
Gregory~F Lawler.
\newblock Topics in loop measures and the loop-erased walk.
\newblock {\em Probability Surveys}, 15:28--101, 2018.

\bibitem{Ricci}
Carfora, M., Dappiaggi, C., Drago, N., and Rinaldi, P. .
\newblock \emph{Ricci Flow from the Renormalization of Nonlinear Sigma Models in the Framework of Euclidean Algebraic Quantum Field Theory.} 
\newblock arXiv preprint arXiv:1809.07652.

\bibitem{Dangquillen}
N.V. Dang.
\newblock Renormalization of determinant lines in quantum field theory.
\newblock 2019.

\bibitem{Glimm}
J.~Glimm and A.~Jaffe.
\newblock {\em Quantum Physics, A Functional Integral Point of View}.
\newblock Springer, New York, 1981.

\bibitem{Taylor-81}
M.~E. Taylor.
\newblock {\em Pseudodifferential Operators}.
\newblock Princeton University Press, Princeton, 1981.


\bibitem{Taylor2}
M.~E. Taylor.
\newblock {\em Partial Differential Equations II}.
\newblock Springer, 2013.


\bibitem{DyZwscatt}
Semyon Dyatlov and Maciej Zworski.
\newblock Mathematical theory of scattering resonances.
\newblock AMS Graduate Studies in Mathematics 200, 2019.

\bibitem{BGV}
N.~Berline, E.~Getzler, and M.~Vergne.
\newblock {\em Heat kernels and Dirac operators}.
\newblock Springer Verlag, Berlin, 2004.

\bibitem{DG75}
Johannes~J Duistermaat and Victor~W Guillemin.
\newblock The spectrum of positive elliptic operators and periodic
  bicharacteristics.
\newblock {\em Inventiones mathematicae}, 29(1):39--79, 1975.

\bibitem{guillemin1977lectures}
Victor Guillemin.
\newblock Lectures on spectral theory of elliptic operators.
\newblock {\em Duke Mathematical Journal}, 44(3):485--517, 1977.

\bibitem{BPP}
Robert Brooks, Peter Perry, and Peter Petersen.
\newblock Compactness and finiteness theorems for isospectral manifolds.
\newblock {\em J. reine angew. Math}, 426:67--89, 1992.

\bibitem{simonpphi2}
Simon, Barry. 
\newblock $P(\varphi)_2$ Euclidean (Quantum) Field Theory.
\newblock Princeton University Press, 2015.

\bibitem{simonfunct}
Simon, Barry.
\newblock Functional integration and quantum physics.
\newblock Vol. 86. Academic press, 1979.

\bibitem{FrohlichSokal}
Fernández, Roberto, J\"urg Fr\"ohlich, and Alan D. Sokal. 
\newblock Random walks, critical phenomena, and triviality in quantum field theory. 
\newblock Springer, 2013.

\bibitem{ammann2002einstein}
Bernd Ammann and Christian B{\"a}r.
\newblock The Einstein-Hilbert action as a spectral action.
\newblock In {\em Noncommutative Geometry and the Standard Model of Elementary
  Particle Physics}, pages 75--108. Springer, 2002.

\bibitem{andersonisospec}
Michael~T Anderson.
\newblock Remarks on the compactness of isospectral sets in low dimensions.
\newblock {\em Duke Mathematical Journal}, 63(3):699--711, 1991.

\bibitem{sharafutdinov2009local}
Vladimir~A Sharafutdinov.
\newblock Local audibility of a hyperbolic metric.
\newblock {\em Siberian Mathematical Journal}, 50(5):929, 2009.

\bibitem{croke1998}
Christopher Croke and Vladimir Sharafutdinov.
\newblock Spectral rigidity of a compact negatively curved manifold.
\newblock {\em Topology}, 37(6):1265--1273, 1998.

\bibitem{LefeuvreM2}
Thibault Lefeuvre.
\newblock {\em Tensor Tomography for Surfaces}.
\newblock 2018.
\newblock Master thesis.

\bibitem{croke2000local}
Christopher Croke, Nurlan Dairbekov, and Vladimir Sharafutdinov.
\newblock Local boundary rigidity of a compact riemannian manifold with
  curvature bounded above.
\newblock {\em Transactions of the American Mathematical Society},
  352(9):3937--3956, 2000.

\bibitem{OPS}
Brad Osgood, Ralph Phillips, and Peter Sarnak.
\newblock Compact isospectral sets of surfaces.
\newblock {\em Journal of functional analysis}, 80(1):212--234, 1988.

\bibitem{dang2018fried}
Nguyen~Viet Dang, Colin Guillarmou, Gabriel Rivi{\`e}re, and Shu Shen.
\newblock Fried conjecture in small dimensions.
\newblock {\em arXiv preprint arXiv:1807.01189}, 2018.

\bibitem{Boileau}
Boileau, Michel. 
\newblock Lectures on Cheeger-Gromov Theory of Riemannian manifolds
\newblock Summer School on Geometry and Topology of 3-manifolds, ICTP TRIESTE June 2005. 

\bibitem{PetersCheegerfiniteness}
Peters, Stefan. 
\newblock Cheeger's finiteness theorem for diffeomorphism classes of Riemannian manifolds.
\newblock Journal f\"ur die reine und angewandte Mathematik 349 (1984): 77-82

\bibitem{Kasue}
Kasue, Atsushi. 
\newblock A convergence theorem for Riemannian manifolds and some applications.
\newblock Nagoya Mathematical Journal 114 (1989): 21-51.

\bibitem{Simon-traceideals}
Barry Simon.
\newblock Trace ideals and their applications, volume 120 of mathematical
  surveys and monographs.
\newblock {\em American Mathematical Society, Providence, RI,}, 2005.


\bibitem{PSU15}
Paternain, Gabriel P., Mikko Salo, and Gunther Uhlmann. 
\newblock Invariant distributions, Beurling transforms and tensor tomography in higher dimensions.
\newblock Mathematische Annalen 363.1-2 (2015): 305-362.

\bibitem{DeLlaveMarcoMoriyon}
Rafael de~la Llave, Jose~Manuel Marco, and Roberto Moriy{\'o}n.
\newblock Canonical perturbation theory of Anosov systems and regularity
  results for the Livsic cohomology equation.
\newblock {\em Annals of Mathematics}, 123(3):537--611, 1986.

\bibitem{Zhou}
Zhou, Gengqiang.
\newblock Compactness of isospectral compact manifolds with bounded curvatures.
\newblock Pacific journal of mathematics, 1997, vol. 181, no 1, p. 187-200.

\end{thebibliography}
\end{document}